\documentclass[runningheads,a4paper]{llncs}

\usepackage{amssymb}
\usepackage{amsmath}
\usepackage{graphicx}
\usepackage{algpseudocode}
\usepackage[T1]{fontenc}
\usepackage[american]{babel}
\usepackage[utf8x]{inputenc}
\usepackage{dsfont}
\usepackage{todonotes}
\usepackage{xspace}
\usepackage{url}
\usepackage{wrapfig}
\usepackage{subfigure}
\usepackage{booktabs}
\usepackage{longtable}
\usepackage[pdfborder={0 0 0}]{hyperref} 

\graphicspath{{./fig/}}

\let\emptyset\varnothing

\newcommand{\old}[1]{{}}
\newcommand{\ignore}[1]{}

\newcommand{\keywords}[1]{\par\addvspace\baselineskip
\noindent\keywordname\enspace\ignorespaces#1}

\newcommand{\MinBound}{the MPP\xspace}

\newcommand{\bbC}{\mathcal{C}\xspace}

\begin{document}

\title{Computing Nonsimple Polygons\\ of Minimum Perimeter}

\author{S\'andor P.\ Fekete\inst{1}
   \and Andreas Haas\inst{1}
   \and Michael Hemmer\inst{1}
   \and Michael Hoffmann\inst{2}
   \and Irina Kostitsyna\inst{3}
   \and Dominik Krupke\inst{1}
   \and Florian Maurer\inst{1}
   \and Joseph S.\ B.\ Mitchell\inst{4}
   \and Arne Schmidt\inst{1}
   \and Christiane Schmidt\inst{5}
   \and Julian Troegel\inst{1}
}
\authorrunning{Fekete et al.}

\institute{
  TU Braunschweig, Germany.
  \and ETH Zurich, Switzerland.
  \and TU Eindhoven, the Netherlands.
  \and Stony Brook University, USA.
  \and Link\"oping University, Sweden.
}

\maketitle

\begin{abstract}
We provide exact and approximation methods for solving a geometric relaxation
of the Traveling Salesman Problem (TSP) that occurs in curve reconstruction:
for a given set of vertices in the plane, the problem Minimum Perimeter Polygon
(MPP) asks for a (not necessarily simply connected) polygon with shortest possible
boundary length. Even though the closely related problem of finding a minimum cycle
cover is polynomially solvable by matching techniques, we prove how the
topological structure of a polygon leads to NP-hardness of the MPP. On the
positive side, we show how to achieve a constant-factor approximation.

When trying to solve MPP instances to provable optimality by means of integer
programming, an additional difficulty compared to the TSP is the fact that only
a subset of subtour constraints is valid, depending not on combinatorics, but
on geometry. We overcome this difficulty by establishing and exploiting
additional geometric properties. This allows us to reliably solve a wide range
of benchmark instances with up to 600 vertices within reasonable time on a
standard machine. We also show that using a natural geometry-based
sparsification yields results that are on average within 0.5\% of the optimum.
\end{abstract}



\section{Introduction}
\label{sec:introduction}

For a given set $V$ of points in the plane, the Minimum Perimeter Polygon (MPP) asks for a
polygon $P$ with vertex set $V$ that has minimimum possible boundary length. An optimal solution
may not be simply connected, so we are faced with a geometric relaxation of the Traveling Salesman Problem (TSP).

The TSP is one of the classic problems of Combinatorial Optimization.
NP-hard even in special cases of geometric instances (such as grid graphs), it has served as one of the 
prototypical testgrounds for developing outstanding algorithmic approaches.
These include constant-factor approximation methods (such as Christofides' 3/2 approximation~\cite{christofides}
in the presence of triangle inequality, or Arora's~\cite{arora} and Mitchell's~\cite{mitchell}
polynomial-time approximation schemes for geometric instances), as well as exact methods
(such as Gr\"otschel's optimal solution to a 120-city instance~\cite{Groetschel1980a} or the award-winning
work by Applegate, Bixby, Chvatal and Cook~\cite{abcc} for solving a 13509-city instance within 10 years of CPU time.) 
The well-established benchmark library TSPLIB~\cite{reinelt1991tsplib} of TSP instances has become so
widely accepted that it is used as a benchmark for a large variety of other optimization problems.
See the books~\cite{tsp-book,tsp-book2} for an overview of various aspects of the TSP
and the books~\cite{tspstudy,tsppursuit} for more details on exact optimization.

\begin{figure}[t!hp]
        \begin{center}
                \vspace*{-0.5cm}
                \includegraphics[width=1.0\textwidth]{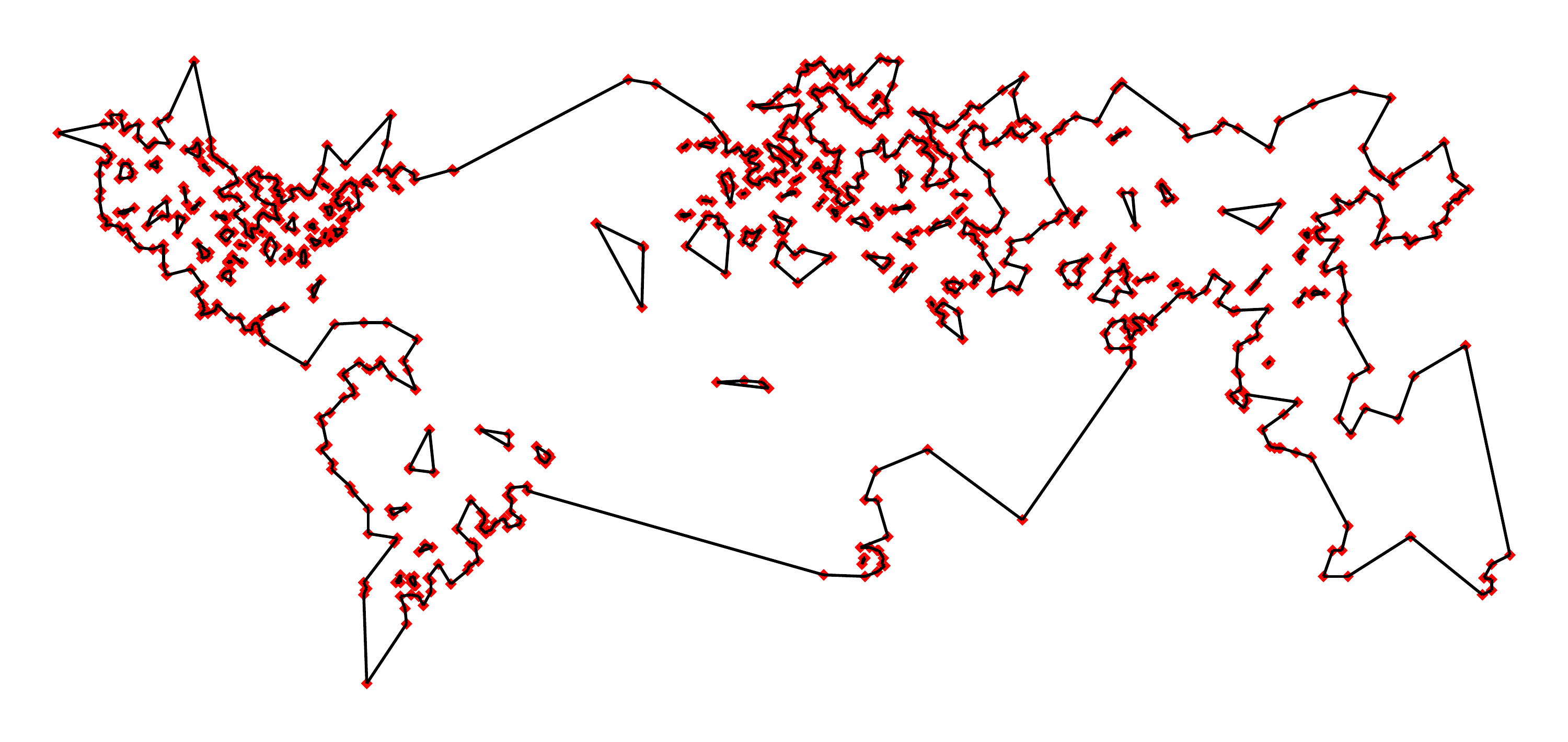}
                \vspace*{-.7cm}
                \caption{A Minimum Perimeter Polygon for an instance with 960 vertices.}
                \vspace*{-.7cm}
                \label{fig:960instance}
        \end{center}
\end{figure}

From a geometric point of view, the TSP asks for the shortest polygonal chain through a given set of 
vertices in the plane; as a consequence of triangle inequality, the result is always a simple polygon of minimum
perimeter. Because of the fundamental role of polygons in geometry, this has made the study
of TSP solutions interesting for a wide range of geometric applications. One such context is geometric shape reconstruction,
where the objective is to re-compute the original curve from a given set of sample points; see Giesen~\cite{giesen},
Althaus and Mehlhorn~\cite{altmehl} or
Dey, Mehlhorn and Ramos~\cite{dots} for specific examples. However, this only makes sense when the original shape
is known to be simply connected, i.e., bounded by a single closed curve. More generally, a shape may be
multiply connected, with interior boundaries surrounding holes. In that case, computing a simple polygon
does not yield the desired answer. Instead, the solution may be a Minimum Perimeter Polygon (MPP):
given a set $V$ of points in the plane, find a not necessarily simple polygon $P$ with vertex set $V$,
such that the boundary of $P$ has smallest possible length\footnote{Note that we exclude degenerate holes that consist
of only one or two vertices.} See Figure~\ref{fig:960instance} for an optimal solution
of an instance with 960 points; this also illustrates the possibly intricate structure of an MPP.

While the problem MPP\footnote{For simplicity, we will also refer to the problem of computing an MPP as ``the MPP''.}
asks for a cycle cover of the given set of vertices (as opposed to the single cycle required by the TSP), 
it is important to note that even the more general geometry of a polygon with holes 
imposes some topological constraints on the structure of boundary cycles; as a consequence,
an optimal 2-factor (a minimum-weight cycle cover of the vertices, which can be computed in
polynomial time) may not yield a feasible solution. 
Fekete et al.~\cite{bff+-corsct-15} gave a generic
integer program for the MPP (and other related problems) that was able to yield optimal solutions for instances up to 50 vertices.
However, the main challenges were left unresolved. 
What is the complexity of computing an MPP? Is it possible to 
develop constant-factor approximation algorithms? And how can we compute provably optimal 
solutions for instances of relevant size?


\subsection*{Our Results}

In this paper, we resolve the main open problems related to the MPP.

\begin{itemize}
	\item
	We prove that \MinBound is NP-hard. This shows that despite of the relationship to the polynomially solvable problem of finding a minimum 2-factor, dealing with the topological structure of the involved cycles is computationally difficult.
	\item
	We give a 3-approximation algorithm.
	\item 
	We provide a general IP formulation with $O(n^2)$ variables to ensure a valid polygonal arrangement for \MinBound.
	\item
	We add additional cuts to reduce significantly the number of cuts needed to eliminate outer components and  holes in holes, leading to a practically useful formulation. 
	\item
	We present experimental results for \MinBound, solving instances with up to 1000 points in the plane to provable optimality within 30 minutes of CPU time.
	\item
	We also consider a fast heuristic that is based on geometric structure, restricting the edge set to the Delaunay triangulation.
	Experiments on structured random point sets show that solutions are on average only about 0.5\% worse than the optimum, with vastly superior runtimes.
\end{itemize}

\section{Complexity}
\label{sec:nphard}

\begin{theorem}
\label{th:nphard}
The MPP problem is NP-hard.
\end{theorem}


\begin{proof}
The proof is based on a reduction from the Minimum Vertex Cover problem for planar graphs:
for an undirected planar graph $G=(V,E)$ and a constant $k$, decide whether
there exists a subset of vertices $V'\subset V$ of size $k=|V'|$ such that for
any edge $(u,v)\in E$, either $u\in V'$ or $v\in V'$.
Given an instance $I_{\text{MVC}}$ of the Minimum Vertex Cover problem we
construct an instance $I_{\text{MB}}$ of the \MinBound problem such that
$I_{\text{MB}}$ has a solution if and only if $I_{\text{MVC}}$ has a solution. 
Given a planar graph $G$, we replace its vertices with vertex gadgets, connect
them with edge gadgets, and add three points at the vertices of a large
triangle enclosing the construction. The triangle will delimit the outer
boundary of the polygon in the instance of the \MinBound problem, and the
vertex and edge gadgets will enforce a choice of cycles covering the points
that will form the holes of the polygon.

\paragraph{Vertex gadget.}
The vertex gadget consists of four points (refer to Figure~\ref{fig:vertex-gadget}). The top three points are always connected by a cycle. If the fourth point $p$ is in the same cycle, that represents putting the corresponding vertex in subset $V'$. The cycle's length is $3\varepsilon$ if the vertex is not in $V'$, and $2b+2\varepsilon$ if the vertex is in $V'$.

\begin{figure}[ht]
\centering
\includegraphics[page=2]{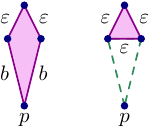}
\caption{Vertex gadget. Left: vertex $\in V'$, total length is $2b+2\varepsilon$; right: vertex $\notin V'$, total length is $3\varepsilon$.}
\label{fig:vertex-gadget}
\end{figure}

\paragraph{Edge gadget.}
The edge gadget consists of a repeating pattern of four points forming a rhombus (refer to Figure~\ref{fig:edge-gadget}). Let some edge gadget consist of $r$ rhombi. There are three ways of covering all the points except for, possibly, the two outermost points, with cycles of total length not greater than $2ra+r\varepsilon$ (see Figure~\ref{fig:edge-gadget} (a-c)). This will leave either the leftmost point, either the rightmost point, or both, the leftmost and the rightmost points, uncovered by the cycles. If we require both outermost points to be covered by the cycles, their total length will be at least $2(r+1)a+(r-1)\varepsilon$ (see Figure~\ref{fig:edge-gadget} (d)). The points of the edge gadget could potentially be covered by a path of length $2ra+r\varepsilon$ (see Figure~\ref{fig:edge-gadget} (e)) that closes into a cycle through other gadgets. To prevent this situation we will add triplets of points that will form small holes in the middle of each face of $G$. A cycle that passes through an edge gadget would enclose at least one face of graph $G$, thus would enclose another hole.

\begin{figure}[ht]
\centering
\includegraphics[page=1]{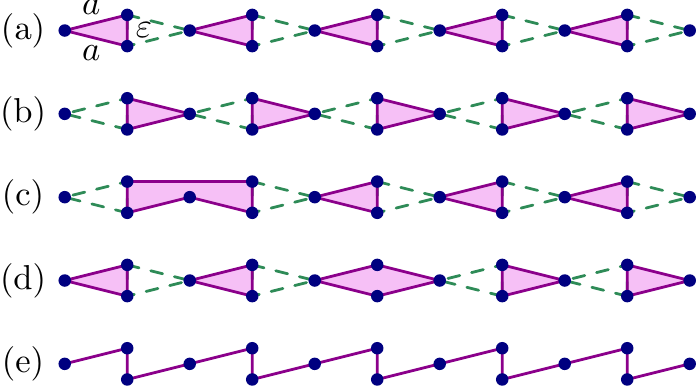}
\caption{Edge gadget. (a)--(c) the gadget is covered by cycles of total length $\approx 10a+5\varepsilon$; (d) total length $12a+4\varepsilon$; (e) the gadget is covered by a path of total length $10a+5\varepsilon$.}
\label{fig:edge-gadget}
\end{figure}

\paragraph{Split gadget.}
Split gadget (refer to Figure~\ref{fig:split-gadget}) multiplies the connection to a vertex gadget, thus allowing us to connect one vertex gadget to multiple vertex gadgets. If point $p$ is covered by the vertex gadget, all the points, including points $p_1$ and $p_2$, of the split gadget can be covered by cycles of total length $16a+11\varepsilon$. If point $p$ is not covered by the vertex gadget, $p$ and all the points of the split gadget, except for $p_1$ and $p_2$, can be covered by cycles of total length $16+11\varepsilon$. Notice, that the cycles can only consist of the edges that are shown in the figure (with solid or dashed lines). There is always the same number of edges used in any collection of cycles that cover the same number of points. Therefore, if some cycle contains an edge that is longer than $a$, the other edges in the cycles will have to be shorter to compensate for the extra length. By a simple case distinction one can show that there is no collection of cycles of length not greater than $16+11\varepsilon$ that covers the same points of the split gadget and that uses any edge that is not shown in Figure~\ref{fig:split-gadget}.

If we require the split gadget to cover points $p_1$ and $p_2$ when point $p$ is not covered by the vertex gadget, the total length of the cycles will be at least $18a+10\varepsilon$ (see Figure~\ref{fig:split-gadget-bad}).

\begin{figure}[ht]
\centering
\includegraphics[page=3,width=0.48\textwidth]{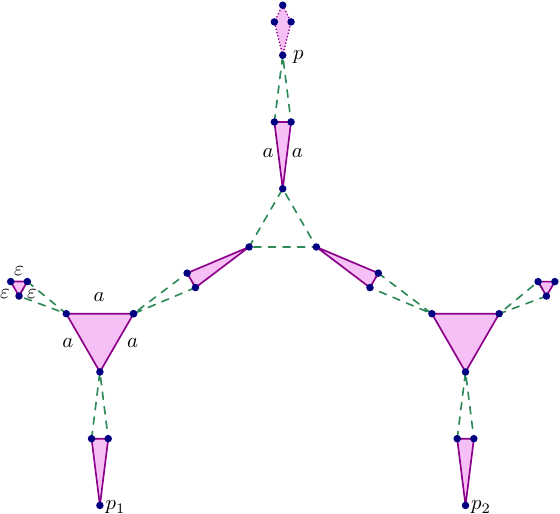}
\hfill
\includegraphics[page=4,width=0.48\textwidth]{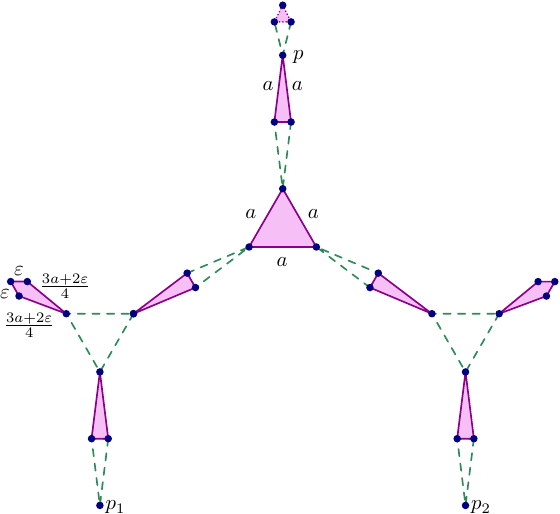}
\caption{Split gadget. Left: vertex $\in V'$, total length is $16a+11\varepsilon$; right: vertex $\notin V'$, total length is $16a+11\varepsilon$.}
\label{fig:split-gadget}
\end{figure}

\begin{figure}[ht]
\centering
\includegraphics[page=5,width=0.48\textwidth]{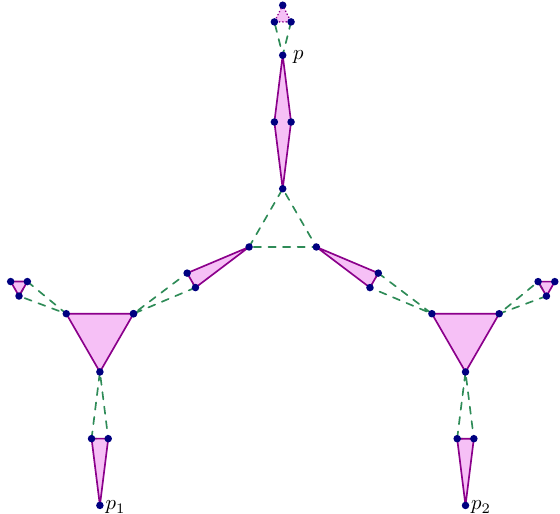}
\caption{Split gadget: points $p$, $p_1$, and $p_2$ are covered by cycles of total length $18a+10\varepsilon$.}
\label{fig:split-gadget-bad}
\end{figure}

To summarize, given an embedding of planar graph $G=(V,E)$ with $n$ vertices and $m$ edges, we construct an instance of \MinBound problem by replacing the vertices of the graph with the vertex gadgets, attaching $deg(v)-1$ split gadgets (where $deg(v)$ denotes the degree of vertex $v$) to the corresponding vertex gadget of every vertex $v$, and connecting the vertex gadgets by edge gadgets (see Figure~\ref{fig:np-hard}). We enclose the construction in a triangle of a very large size, that will be the outer boundary of the polygon. Let the perimeter of the triangle $T \gg$ than the diameter of $G$. The cycles covering the points of the gadgets will be the holes in the polygon. Moreover, to every face of $G$ we add triplets of points forming cycles of a very small length $\ll\varepsilon$. This will eliminate any possibility of passing through edge gadgets with a single cycle.

\begin{figure}[ht]
\centering
\includegraphics[page=6]{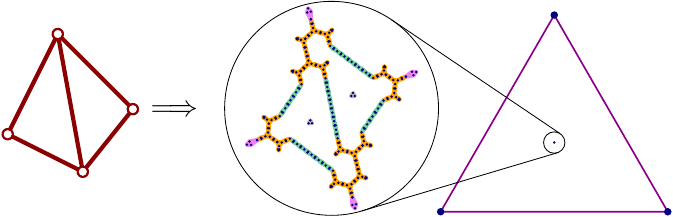}
\caption{Given a planar graph $G$, we construct an instance of \MinBound problem. Highlighted in violet are the vertex gadgets, in orange are the split gadgets, in green are the edge gadgets, and in gray are the extra holes in the middle of the faces of $G$.}
\label{fig:np-hard}
\end{figure}

The number of vertex gadgets used in the construction is $n$, and the number of split gadgets is $\sum_{v\in V} deg(v)-n=2m-n$. Let the number of rhombi used in all the edge gadgets be $r$, and let the total length of the extra holes in the middle of the faces of $G$ be $\varepsilon$. Then the instance of the \MinBound problem will ask whether there exists a polygon of perimeter not greater than
\begin{multline*}
L=T+k(2b+2\varepsilon)+(n-k)3\varepsilon+(2m-n)(16a+11\varepsilon)+2ra+r\varepsilon+\varepsilon=\\
(2b-\varepsilon)k+T+2(16 m-8 n+r)a+(22 m-8 n+r+1)\varepsilon\,.
\end{multline*}
Let $d$ be the length of the shortest edge. Choose $a$, $b$, and $\varepsilon$, such that $\varepsilon\ll b\ll a\ll d$. 
Then, there is a polygon with perimeter at most $L$ if and only if there is a vertex cover of size $k$ in the instance of Minimum Vertex Cover problem.

Let $V'$ be a vertex cover of size $k$ of $G=(V,E)$. Then, by selecting the corresponding vertex gadgets to cover points $p$, and propagating the construction of cycles along the split and edge gadgets, we get a polygon of perimeter $L$.

Let there exist a polygon $P$ with perimeter not greater than $T+2(16 m-8n+r)a+2kb+(22 m-8 n+r-k+1)\varepsilon$. By construction, the outer boundary of $P$ will be the triangle of perimeter $T$. Suppose there are more than $k$ vertex gadgets that are covering the corresponding points $p$. Then the perimeter of $P$ has to be greater than $T+2(16 m-8n+r)a+2kb+(22 m-8 n+r-k+1)\varepsilon$, as the third term (of variable $b$) of the perimeter expression dominates the fourth term (of variable $\varepsilon$). Thus, there has to be not more than $k$ variable gadgets that cover the corresponding points $p$. Every edge gadget has to have one of the end-points covered by the vertex gadgets (through split gadgets). Otherwise, the second term of the expression of the polygon perimeter would be greater. Therefore, the polygon corresponds to a vertex cover of the Minimum Vertex Cover instance of size not greater than $k$.
\qed
\end{proof}

\section{Approximation}
\label{sec:approx}

In this section we show that \MinBound can be approximated within
a factor of 3. 

\begin{theorem}
\label{thm:approx}
There exists a polynomial time $3$-approximation for \MinBound.
\end{theorem}

\begin{proof}
  Let $OPT$ be the length of an optimal
  solution of \MinBound and $APX$ the length of the approximation that
  our algorithm will compute for the given set, $V$, of $n$ points in
  the plane.  
We compute the convex hull, $CH(V)$, of the input set; this takes time
$O(n \log h)$, where $h$ is the number of vertices of the convex hull.
Note that the perimeter, $|CH(V)|$, of the convex hull is a lower
bound on the length of an optimal solution ($OPT\geq |CH(V)|$), since
the outer boundary of any feasible solution polygon must enclose all
points of $V$, and the convex hull is the minimum-perimeter enclosure
of~$V$.

Let $U\subseteq V$ be the input points interior to $CH(V)$.  If
$U=\emptyset$, then the optimal solution is given by the convex hull.
If $|U|\leq 2$, we claim that an optimal solution is a simple
(nonconvex) polygon, with no holes, on the set $V$, given by the TSP
tour on $V$; since $|U|=2$ is a constant, it is easy to compute the
optimal solution in polynomial time, by trying all possible ways of
inserting the points of $U$ into the cycle of the points of $V$ that
lie on the boundary of the convex hull, $CH(V)$.

Thus, assume now that $|U|\geq 3$.  We compute a minimum-weight
2-factor, denoted by $\gamma(U)$, on $U$, which is done in polynomial-time by
standard methods~\cite{bills-book}.  Now, $\gamma(U)$ consists of a
set of disjoint simple polygonal curves having vertex set $U$; the
curves can be nested, with possibly many levels of nesting.  We let
$F$ denote the directed {\em nesting forest} whose nodes are the
cycles (connected components) of $\gamma(U)$ and whose directed edges
indicate nesting (containment) of one cycle within another; refer to
Figure~\ref{fig:nesting-forest}.  Since an optimal solution consists
of a 2-factor (an outer cycle, together with a set of cycles, one per
hole of the optimal polygon), we know that $OPT\geq |\gamma(U)|$.  (In
an optimal solution, the nesting forest corresponding to the set of
cycles covering all of $V$ (not just the points $U$ interior to
$CH(V)$) is simply a single tree that is a star: a root node
corresponding to the outer cycle, and a set of children adjacent to
the root node, corresponding to the boundaries of the holes of the
optimal polygon.)  If the nesting forest $F$ for our optimal 2-factor
is a set of isolated nodes (i.e., there is no nesting among the cycles
of the optimal 2-factor on $U$), then our algorithm outputs a polygon
with holes whose outer boundary is the boundary of the convex hull,
$CH(V)$, and whose holes are the (disjoint) polygons given by the
cycles of $\gamma(U)$.  (In this case, the total weight of our
solution is equal to $|CH(V)|+|\gamma(U)| \leq 2\cdot OPT$.)

\begin{figure}
\centering
\includegraphics[width=0.80\columnwidth]{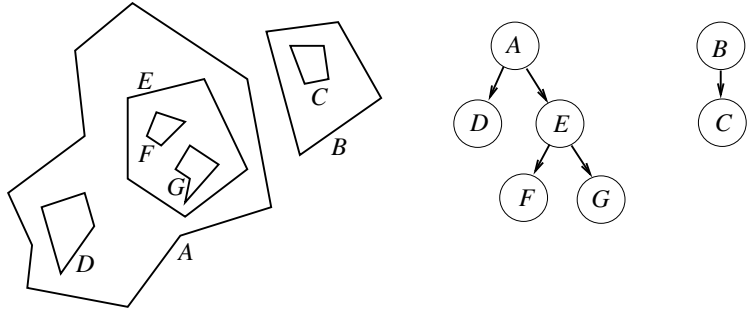}
\caption{A 2-factor (left) and its corresponding nesting forest (right).}
\label{fig:nesting-forest}
\end{figure}

Assume now that $F$ has at least one nontrivial tree.  We describe a
two-phase process that transforms the set of cycles corresponding to
$F$ into a set of pairwise-disjoint cycles, each defining a simple
polygon interior to $CH(V)$, with no nesting -- the resulting simple
polygons are disjoint, each having at least 3 vertices (from $U\subset
V$),

Phase 1 of the process transforms the cycles $\gamma(U)$ to a set of
polygonal cycles that define {\em weakly simple} polygons whose
interiors are pairwise disjoint.  (A polygonal cycle $\beta$ defines a
{\em weakly simple} polygon $P_\beta$ if $P_\beta$ is a closed, simply
connected set in the plane with a boundary, $\partial P_\beta$
consisting of a finite union of line segments, whose traversal (e.g.,
while keeping the region $P_\beta$ to one's left) is the
(counterclockwise) cycle $\beta$ (which can have line segments that
are traversed twice, once in each direction).)  The total length of
the cycles at the end of phase 1 is at most 2 times the length of the
original cycles, $\gamma(U)$.  Then, phase 2 of the process transforms
these weakly simple cycles into (strongly) simple cycles that define
disjoint simple polygons interior to $CH(V)$.  Phase 2 only does
shortening operations on the weakly simple cycles; thus, the length of
the resulting simple cycles at the end of phase 2 is at most 2 times
the total length of~$\gamma(U)$.
%
At the end of phase 2, we have a set of disjoint simple polygons within $CH(V)$,
which serve as the holes of the output polygon, whose total
perimeter length is at most $|CH(V)|+2|\gamma(U)| \leq 3\cdot OPT$.

We now describe phase 1.  Let $T$ be a nontrivial tree of $F$.
Associated with $T$ are a set of cycles, one per node.  A node $u$ of
$T$ that has no outgoing edge of $T$ (i.e., $U$ has no children) is a
sink node; it corresponds to a cycle that has no cycle contained
within it.  Let $v$ be a node of $T$ that has at least one child, but
no grandchildren.  (Such a node must exist in a nontrivial tree $T$.)
Then, $v$ corresponds to a cycle (simple polygon) $P_v$, within which
there is one or more disjoint simple polygonal cycles,
$P_{u_1},P_{u_2},\ldots,P_{u_k}$, one for each of the $k\geq 1$
children of $v$.  We describe an operation that replaces $P_v$ with a
new weakly simple polygon, $Q_v$, whose interior is disjoint from
those of $P_{u_1},P_{u_2},\ldots,P_{u_k}$.  Let $e=pq$ ($p,q\in V$) be
any edge of $P_v$; assume that $pq$ is a counterclockwise edge, so
that the interior of $P_v$ lies to the left of the oriented segment
$pq$.  Let $\Gamma$ be a shortest path within $P_v$, from $p$ to $q$,
that has all of the polygons $P_{u_1},P_{u_2},\ldots,P_{u_k}$ to its
right; thus, $\Gamma$ is a ``taut string'' path within $P_v$,
homotopically equivalent to $\partial P_v$, from $p$ to $q$.  (Such a
geodesic path is related to the ``relative convex hull'' of the
polygons $P_{u_1},P_{u_2},\ldots,P_{u_k}$ within $P_v$, which is the
shortest cycle within $P_v$ that encloses all of the polygons; the
difference is that $\Gamma$ is ``anchored'' at the endpoints $p$ and
$q$.)  Note that $\Gamma$ is a polygonal path whose vertices are
either (convex) vertices of the polygons $P_{u_j}$ or (reflex)
vertices of $P_v$.  Consider the closed polygonal walk that starts at
$p$, follows the path $\Gamma$ to $q$, then continues counterclockwise
around the boundary, $\partial P_v$, of $P_v$ until it returns to $p$.
This closed polygonal walk is the counterclockwise traversal of a
weakly simple polygon, $Q_v$, whose interior is disjoint from the
interiors of the polygons $P_{u_1},P_{u_2},\ldots,P_{u_k}$.  Refer to
Figure~\ref{fig:phase1}.  The length of this closed walk (the
counterclockwise traversal of the boundary of $Q_v$) is at most twice
the perimeter of $P_v$, since the path $\Gamma$ has length at most
that of the counterclockwise boundary $\partial P_v$, from $q$ to $p$
(since $\Gamma$ is a homotopically equivalent shortening of this
boundary).  We consider the boundary of $P_v$ to be replaced with the
cycle around the boundary of $Q_v$, and this process has reduced the
degree of nesting in $T$: node $v$ that used to have $k$ children
(leaves of $T$) is now replaced by a node $v'$ corresponding to $Q_v$,
and $v'$ and the $k$ children of $v$ are now all siblings in the
modified tree, $T'$.  If $v$ had a parent, $w$, in $T$, then $v'$ and
the $k$ children of $v$ are now children of $W$; if $v$ had no parent
in $T$ (i.e., it was the root of $T$), then $T$ has been transformed
into a set of $k+1$ cycles, none of which are nested within another
cycle of $\gamma(U)$.  (Each is within the convex hull $CH(V)$, but
there is no other surrounding cycle of $\gamma(U)$.)  We continue this
process of transforming a surrounding parent cycle (node $v$) into a
sibling cycle (node $v'$), until each tree $T$ of $F$ becomes a set of
isolated nodes, and finally $F$ has no edges (there is no nesting).

\begin{figure}
\centering
\includegraphics[width=0.80\columnwidth]{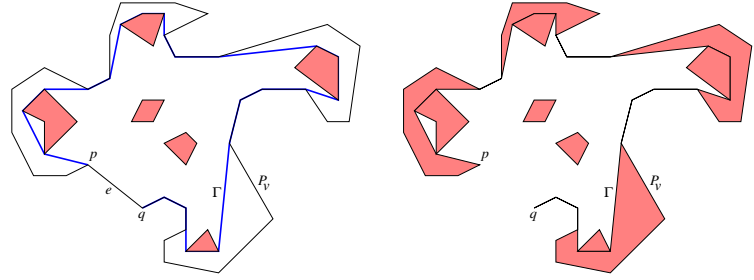}
\caption{Left: The geodesic path $\Gamma$ from $p$ to $q$ within $P_v$, surrounding all of the 
(red) polygons $P_{u_1},P_{u_2},\ldots,P_{u_k}$.  Right: The new weakly simple polygon (now red)
obtained from the traversal of $\Gamma$ and the boundary of $P_v$.}
\label{fig:phase1}
\end{figure}

Phase 2 is a process of local shortening of the cycles/polygons,
$Q_1,Q_2,\ldots,Q_m$, that resulted from phase 1, in order to remove
repeated vertices in the weakly simple cycles, so that cycles become
(strongly) simple.  There are two types of repeated vertices to
resolve: those that are repeated within the same cycle (i.e., repeated
vertices $p$ of a cycle $Q_i$ where $\partial Q_i$ ``pinches'' upon
itself), and those that are repeated across different cycles (i.e.,
vertices $p$ where one cycle is in contact with another, both having
vertex $p$).  

Consider a weakly simple polygon $Q$, and let $p$ be a vertex of $Q$
that is repeated in the cycle specifying the boundary $\partial Q$.
This implies that there are four edges of the (counterclockwise)
cycle, $p_0p$, $pp_1$, $p_2p$, and $pp_3$, incident on $p$, all of which
lie within a halfplane through $p$ (by local optimality).  There are then two subcases:
(i) $p_0,p,p_1$ is a left turn (Figure~\ref{fig:phase2}, left); and (ii) $p_0pp_1$ is
a right turn (Figure~\ref{fig:phase2}, right).
In subcase (i), $p_0p, pp_1$ define a left turn at $p$ (making $p$ locally convex for $Q$),
and $p_2p, pp_3$ define a right turn at $p$ (making $p$ locally
reflex for $Q$).  In this case, we replace the pair of edges $p_0p, pp_1$ with a shorter
polygonal chain, namely the ``taut'' version of this path, from $p_0$
to $p_1$, along a shortest path, $\beta_{0,1}$, among the polygons
$Q_i$, including $Q$, treating them as obstacles.  The taut path $\beta_{0,1}$
consists of left turns only, at (locally convex) vertices of polygons
$Q_i$ ($Q_i\neq Q$) or (locally reflex) vertices of $Q$, where new pinch points of $Q$ are created.  
Refer to Figure~\ref{fig:phase2}, left.
Case (ii) is treated similarly; see Figure~\ref{fig:phase2}, right.
%
%
Thus, resolving one repeated vertex, $p$, of $Q$ can result in the
creation of other repeated vertices of $Q$, or repeated vertices where
two cycles come together (discussed below).  The process is finite,
though, since the total length of all cycles strictly decreases with
each operation; in fact, there can be only a polynomial number of such
adjustments, since each triple $(p_0,p,p_1)$, is resolved at most
once.

\begin{figure}
\centering
\includegraphics[width=0.80\columnwidth]{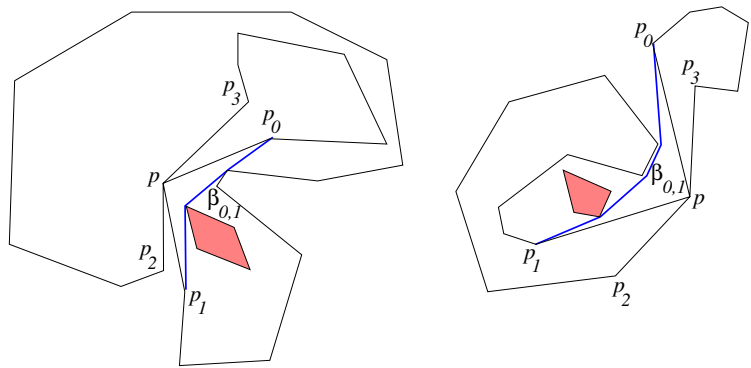}
\caption{Left: Case (i) of the phase 2 shortening process for a pinch point of $Q$.
Right: Case (ii) of the phase 2 shortening process.}
\label{fig:phase2}
\end{figure}

Now consider a vertex $p$ that appears once as a reflex vertex in
$Q_1$ (with incident ccw edges $p_0p$ and $pp_1$) and once as a convex
vertex in $Q_2$ (with incident ccw edges $p_2p$ and $pp_3$).  (Because
cycles resulting after phase 1 are locally shortest, $p$ must be
reflex in one cycle and convex in the other.)  Our local operation in
this case results in a merging of the two cycles $Q_1$ and $Q_2$ into
a single cycle, replacing edges $p_0p$ (of $Q_1$) and $pp_3$ (of
$Q_2$) with the taut shortest path, $\beta_{0,3}$.  As in the process
described above, this replacement can result in new repeated vertices,
as the merged cycle may come into contact with other cycles, or with
itself.

\begin{figure}
\centering
\includegraphics[width=0.40\columnwidth]{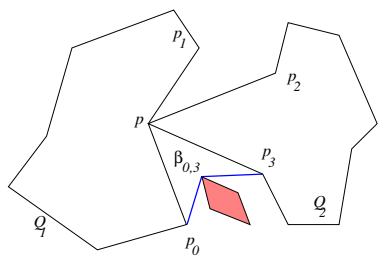}
\caption{The phase 2 shortening process for a point $p$ shared by cycles $Q_1$ and $Q_2$.}
\label{fig:phase2b}
\end{figure}

Finally, the result of phase 2, is a set of disjoint cycles, with no
repeated vertices, defining disjoint simple polygons within $CH(V)$;
these cycles define the holes of the output polygon, whose total
perimeter length is at most that of $CH(V)$, plus twice the lengths of
the cycles $\gamma(U)$ in an optimal 2-factor of the interior points
$U$.  Thus, we obtain a valid solution with objective function at most
3 times optimal.
\end{proof}

\section{IP Formulation}
\label{sec:ip_formulation}

\subsection{Cutting-Plane Approach}
In the following we develop suitable Integer Program (IPs) for solving \MinBound to provable optimality. 
The basic idea is to use a binary variable $x_e\in\{0,1\}$ for any possible edge $e\in E$,
with $x_e=1$ corresponding to making $e$ part of a solution $P$ if an only if $x_e=1$.
This allows it to describe the objective function by $\min\sum_{e \in E} x_e c_e$, where 
$c_e$ is the length of $e$. In addition, we impose a suitable set of linear constraints
on these binary variables, such that they characterize precisely the set of polygons with vertex set $V$.
The challenge is to pick a set of constraints that achieve this in a (relatively) efficient manner.

As it turns out (and is discussed in more detail in Section~\ref{sec:separation}), there is a significant
set of constraints that correspond to eliminating cycles within proper subsets $S\subset V$. Moreover,
there is an exponential number of relevant subsets $S$, making it prohibitive to impose all of these
constraints at once. The fundamental idea of a cutting-plane approach is that much fewer
constraints are necessary for characterizing an optimal solution. 
To this end, only a relatively small subfamily of constraints is initially
considered, leading to a relaxation. As long as solving the current relaxation yields 
a solution that is infeasible for the original problem, violated constraints are added in a piecemeal fashion.

In the following, these constraints 
(which are initially omitted, violated by an optimal solution of the relaxation,
then added to eliminate such infeasible solutions) are called {\em cuts},
as they remove solutions of a relaxation that are infeasible for the MPP.

\subsection{Basic IP}
We start with a basic IP that is enhanced with specific cuts, described in Sections \ref{sec:glueCut}--\ref{sec:hihCut}.
We denote by $E$ the set of all edges between two points of $V$,
$\bbC$ the set of invalid cycles and 
$\delta(v)$ the set of all edges in $E$ that are incident to $v\in V$. 
Then we optimize over the following objective function:
\begin{align}
	\label{objectiveHoles} \text{min} \sum_{e \in E} x_e c_e\,.
\end{align}
This is subject to the following constraints:
\begin{align}
	\label{degree} \forall v\in P: \sum_{e\in \delta(v)} x_e &= 2\,,\\
	\label{invalidCircle} \forall C\in \bbC: \hspace*{0.3cm}  \sum_{e\in C} x_e &\leq |C| - 1\,,\\
	\label{variableRestr}	x_e &\in \{0, 1\}\,.
\end{align}

For the TSP, $\bbC$ is simply the set of {\em all} subtours, making identification and separation straightforward.
This is much harder for the MPP, where a subtour may end up being feasible by forming the boundary of a hole,
but may also be required to connect with other cycles.
Therefore, identifying valid inequalities requires more geometric analysis, such as the following.
If we denote by $CH$ the set of all convex hull points, then a cycle $C$
is invalid if $C$ contains: \begin{enumerate}
	\item \label{prop:gluecut}
		at least one and at most $|CH|- 1$ convex hull points. (See Figure \ref{fig:outerBound}) 
	\item \label{prop:tailcut} 
		all convex hull points but does not enclose all other points. (See Figure \ref{fig:outerComp})
	\item \label{prop:holeinhole} 
		no convex hull point but encloses other points. (See Figure \ref{fig:holeinhole}) 
\end{enumerate} 

\begin{figure}[tb]
	\centering
	\subfigure[Invalid cycle of type \ref{prop:gluecut}\label{fig:outerBound}]
	{\includegraphics[width=0.33\columnwidth]{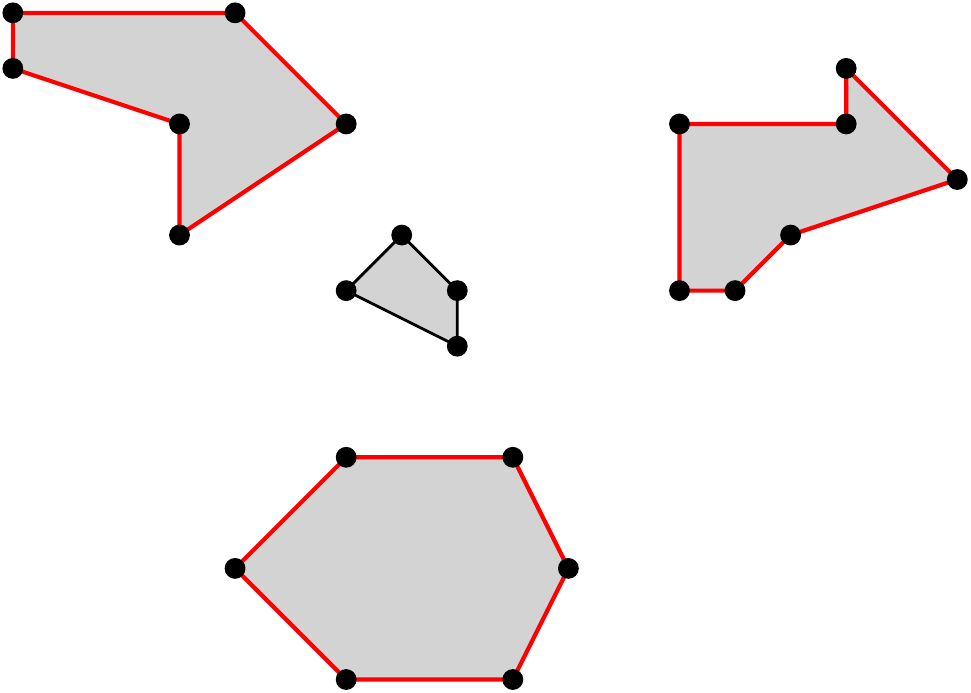}}
	\hfil
	\subfigure[Invalid cycle of type \ref{prop:tailcut}\label{fig:outerComp}]
	{\includegraphics[width=0.22\columnwidth]{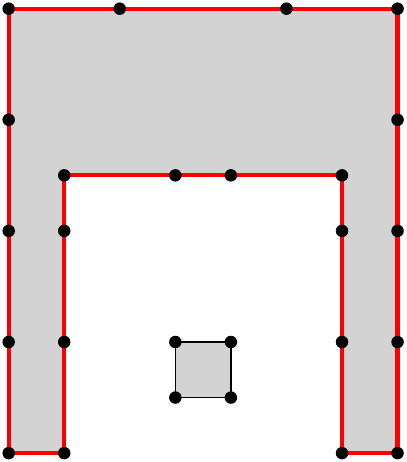}}
	\hfil
	\subfigure[Invalid cycle of type \ref{prop:holeinhole}\label{fig:holeinhole}]
	{\includegraphics[width=0.22\columnwidth]{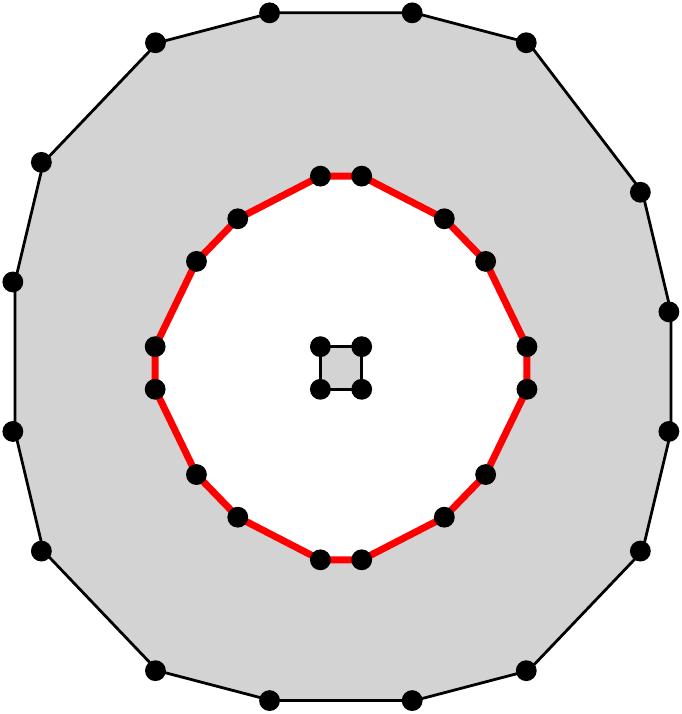}}
	\caption{Examples of invalid cycles (red). Black cycles may be valid.}
	\label{fig:invalidCircles}
\end{figure}

\noindent
By $\bbC_i$ we denote the set of all invalid cycles with property $i$. 
Because there can be exponentially many invalid cycles, we add constraint \eqref{invalidCircle} in separation steps.

\bigskip
\noindent
For an invalid cycle with property \ref{prop:gluecut} we use the equivalent cut constraint
\begin{align}
\label{invalidCircle2} \forall C\in \bbC_1: \hspace*{0.3cm}  \sum_{e\in \delta(C)} x_e \geq 2\,.
\end{align}
By using constraint \eqref{invalidCircle} if $|C| \leq \frac{2n + 1}{3}$ and constraint \eqref{invalidCircle2} otherwise,
where $\delta(C)$ denotes the ``cut'' edges connecting a vertex $v\in C$ with a vertex $v'\not\in C$.
As argued by Pferschy and Stanek \cite{pferschy2014generating}, this technique of {\em dynamic subtour constraints} (DSC)
is useful, as it reduces the number of non-zero coefficients in the constraint matrix.

\ignore{
\begin{theorem}
	With the IP ~\eqref{objectiveHoles}-~\eqref{variableRestr} and the characterization of invalid cycles we can find an optimal solution for \MinBound. 
\end{theorem}

\begin{proof}
With the characterization we ensure an outer hull containing all other points  which then build cycles that contain no other cycle. Hence, our IP gives a solution for \MinBound.
Otherwise, an optimal solution of \MinBound does not violate any constraints of the IP. Thus, it must be a feasible solution of it. Since the optimal solution of \MinBound is the smallest possible and our IP gives a \MinBound solution, we get the optimal solution for \MinBound from our IP.
\end{proof}}

\subsection{Initial Edge Set}
\label{sec:jumpStart}
\ignore{\todo[inline]{This is nothing new! Check Appelgate et al. 12.1.(1,2)}
Since there are many edges which are never be in any solution and thus only slow down the computation of an optimal solution, we came up with the idea to use a small set of edges. We use the edges of a delaunay triangulation, so we only have $O(n)$ instead of $O(n^2)$ edges to consider in our computation. Now we can solve the IP on these edges. When we reached the optimal solution on the delaunay edges, we switch back to the computation with all $O(n^2)$ edges by adding them to the existing IP. 

The advantage of this optimization is that we can identify needed cuts in the separation steps faster which also holds for the IP with all edges.

The disadvantage is that we may not be able to solve anything on delaunay edges since a delaunay triangulation may not be hamiltonian (see Figure \ref{fig:delaunayHamiltonian}).
}


In order to quickly achieve an initial solution, we sparsify the 
$\Theta(n^2)$ input edges to the $O(n)$ edges of the Delaunay Triangulation,
which naturally captures geometric nearest-neighbor properties.
If a solution exists, this yields an upper bound.
This technique has already been applied for the TSP by
J\"unger et al. \cite{junger1995traveling}. 
In theory, this may not yield a feasible solution: a specifically designed example by Dillencourts shows that the Delaunay
triangulation may be non-Hamiltonian~\cite{dillencourt1987non}; this same
example has no feasible solution for \MinBound. We did not observe this behavior in practice.

CPLEX uses this initial solution as an upper bound, quickly allowing it to quickly
discard large solutions in a branch-and-bound manner. 
As described in Section \ref{sec:experiments}, the resulting bounds are quite good
for the MPP.

\section{Separation Techniques}
\label{sec:separation}

\subsection{Pitfalls}
When separating infeasible cycles, the Basic IP may get stuck in an exponential number of iterations, due to the
following issues. (See Figures~\ref{fig:GCwith/out}-\ref{fig:HIHCwith/out} for illustrating examples.)

\begin{description}
	\item[Problem 1:] Multiple outer components containing convex hull points occur that despite the powerful subtour constraints do not get connected because it is cheaper to, e.g., integrate subsets of the interior points.
		Such an instance can be seen in Figure~\ref{fig:GCwith/out}, where we have two equal components with holes.
		Since the two components are separated by a distance greater than the distance between their outer components and their interior points, the outer components start to include point subsets of the holes.
		This results in a potentially exponential number of iterations.
	\item[Problem 2:] Outer components that do not contain convex hull points do not get integrated because we are only allowed to apply a cycle cut on the outer component containing the convex hull points.
		An outer component that does not contain a convex hull point cannot be prohibited as it may become a hole in later iterations.
		See Figure~\ref{fig:TCwith/out} for an example where an exponential number of iterations is needed until the outer components get connected.
	\item[Problem 3:] If holes contain further holes, we are only allowed to apply a cycle cut on the outer hole.
		This outer hole can often cheaply be modified to fulfill the cycle cut but not resolve the holes in the hole.
		An example instance can be seen in Figure~\ref{fig:HIHCwith/out}, where an exponential number of iterations is needed.
\end{description}

\begin{figure}[h!]
	\centering
	\subfigure[]
	{\includegraphics[width=0.18\columnwidth]{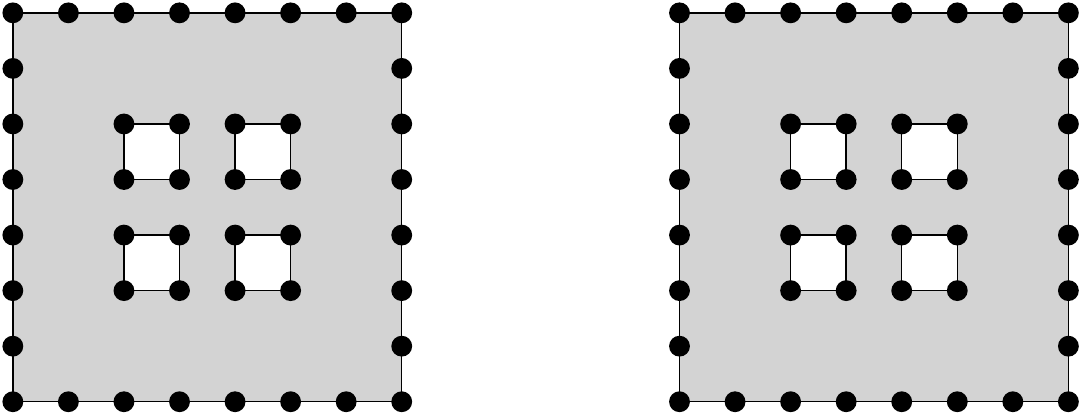}}	
	\hfil
	\subfigure[]
	{\includegraphics[width=0.18\columnwidth]{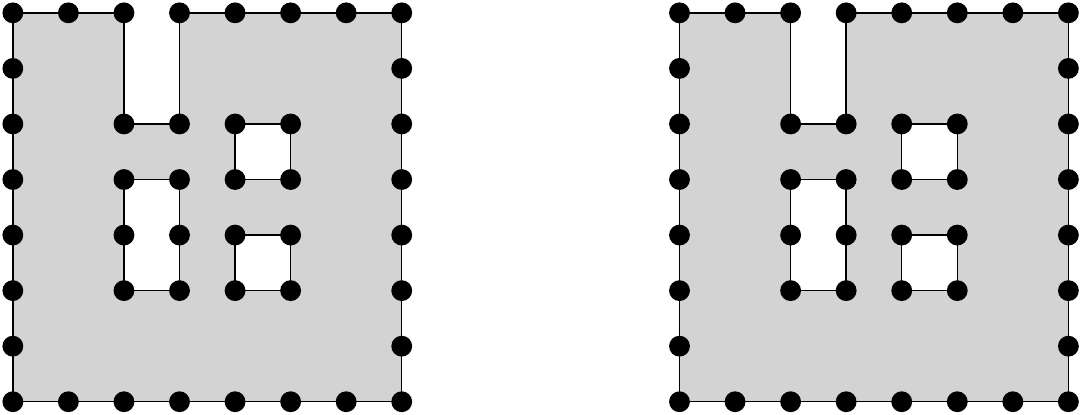}}
	\hfil
	\subfigure[]
	{\includegraphics[width=0.18\columnwidth]{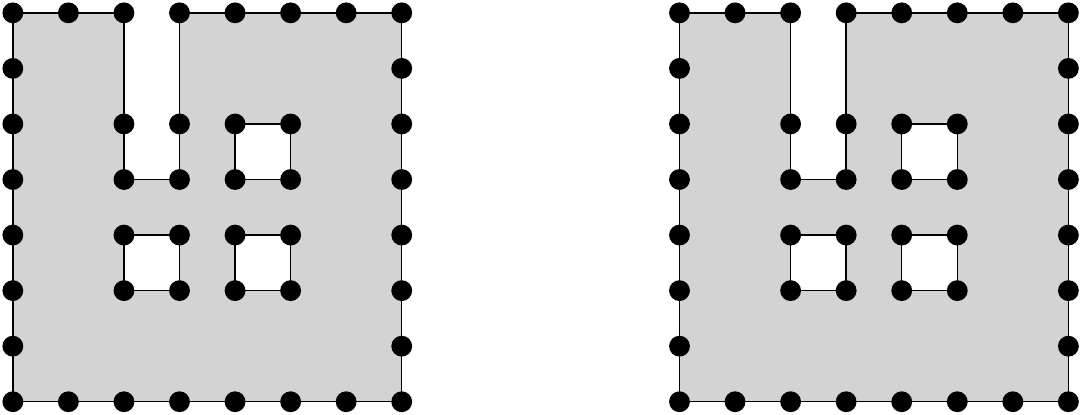}}

	\subfigure[]
	{\includegraphics[width=0.18\columnwidth]{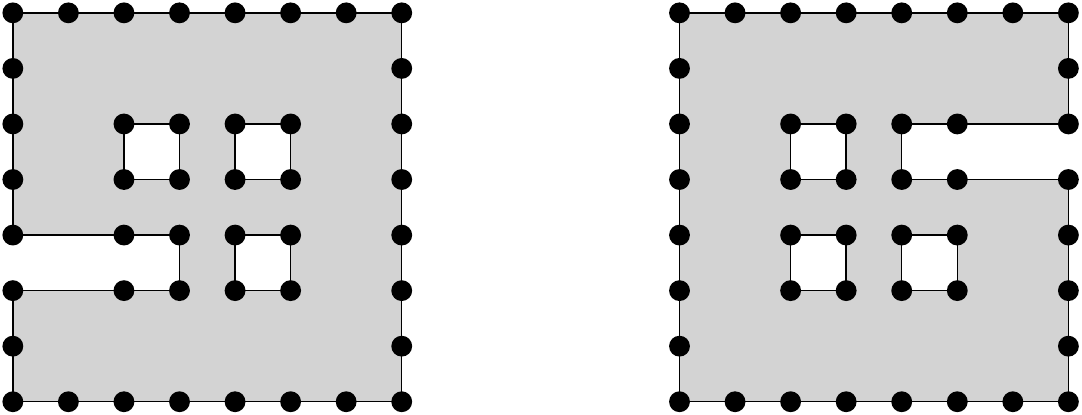}}
	\hfil
	\subfigure[]
	{\includegraphics[width=0.18\columnwidth]{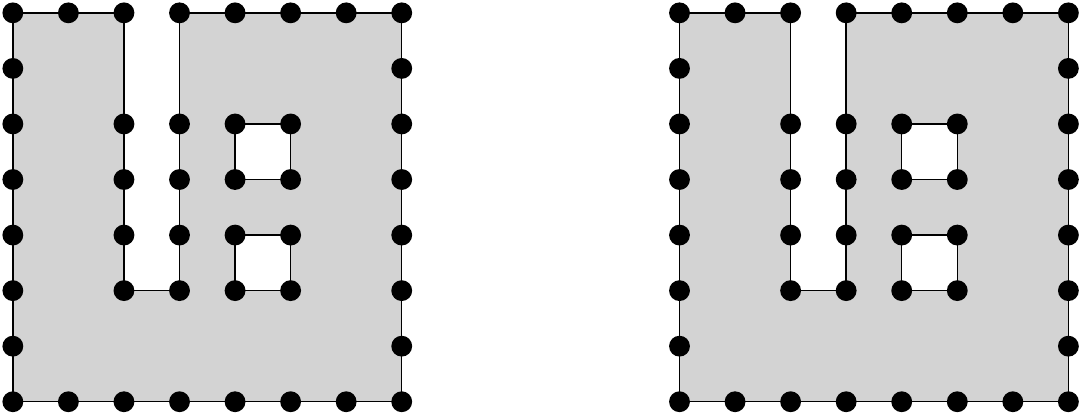}}
	\hfil	
	\subfigure[]
	{\includegraphics[width=0.18\columnwidth]{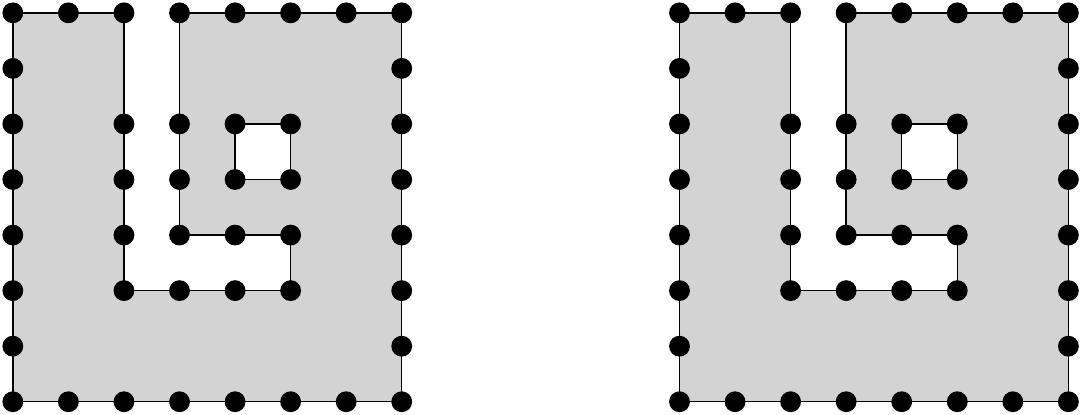}}
	\caption{(a) - (f) show consecutive iterations trying to solve an instance using only constraint \eqref{invalidCircle2}.}
	\label{fig:GCwith/out}
\end{figure}

\begin{figure}[h!]
	\centering
	\subfigure[]
	{\includegraphics[width=0.07\columnwidth,angle=90]{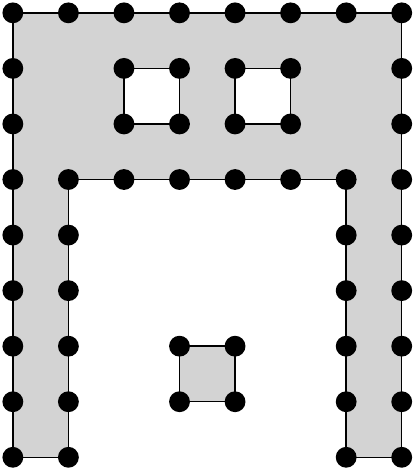}}
	\hfil
	\subfigure[]
	{\includegraphics[width=0.07\columnwidth,angle=90]{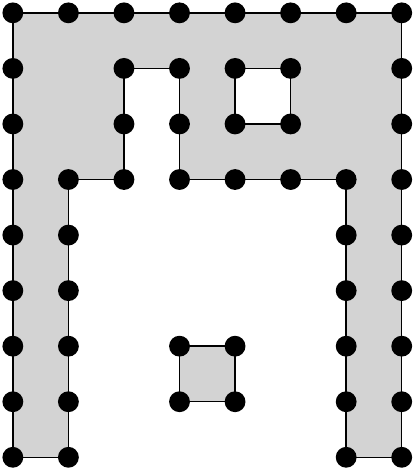}}
	\hfil
	\subfigure[]
	{\includegraphics[width=0.07\columnwidth,angle=90]{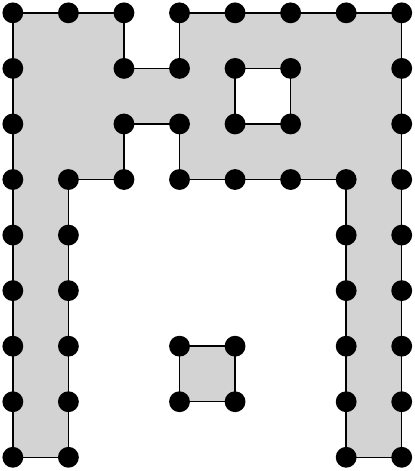}}
	\hfil
	\subfigure[]
	{\includegraphics[width=0.07\columnwidth,angle=90]{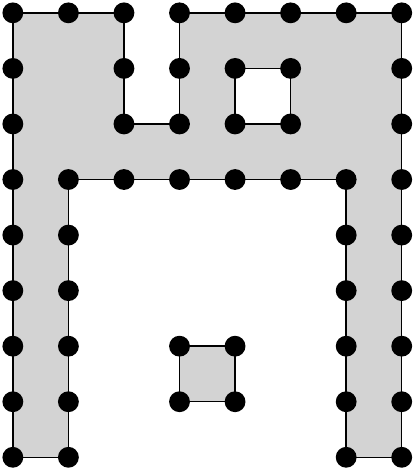}}
	
	\subfigure[]
	{\includegraphics[width=0.07\columnwidth,angle=90]{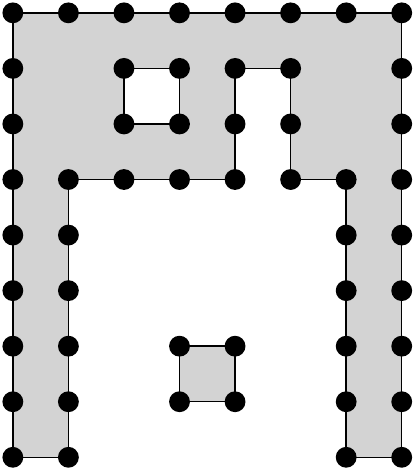}}
	\hfil
	\subfigure[]
	{\includegraphics[width=0.07\columnwidth,angle=90]{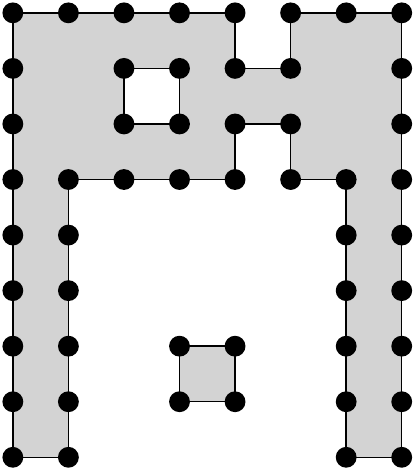}}
	\hfil
	\subfigure[]
	{\includegraphics[width=0.07\columnwidth,angle=90]{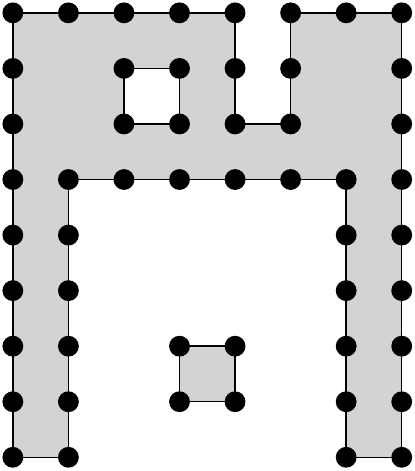}}
	\caption{(a) - (g) show consecutive iterations trying to solve an instance using only constraint \eqref{invalidCircle}.}
	\label{fig:TCwith/out}
\end{figure}

\begin{figure}[h!]
	\centering
	\subfigure[]
	{\includegraphics[width=0.10\columnwidth]{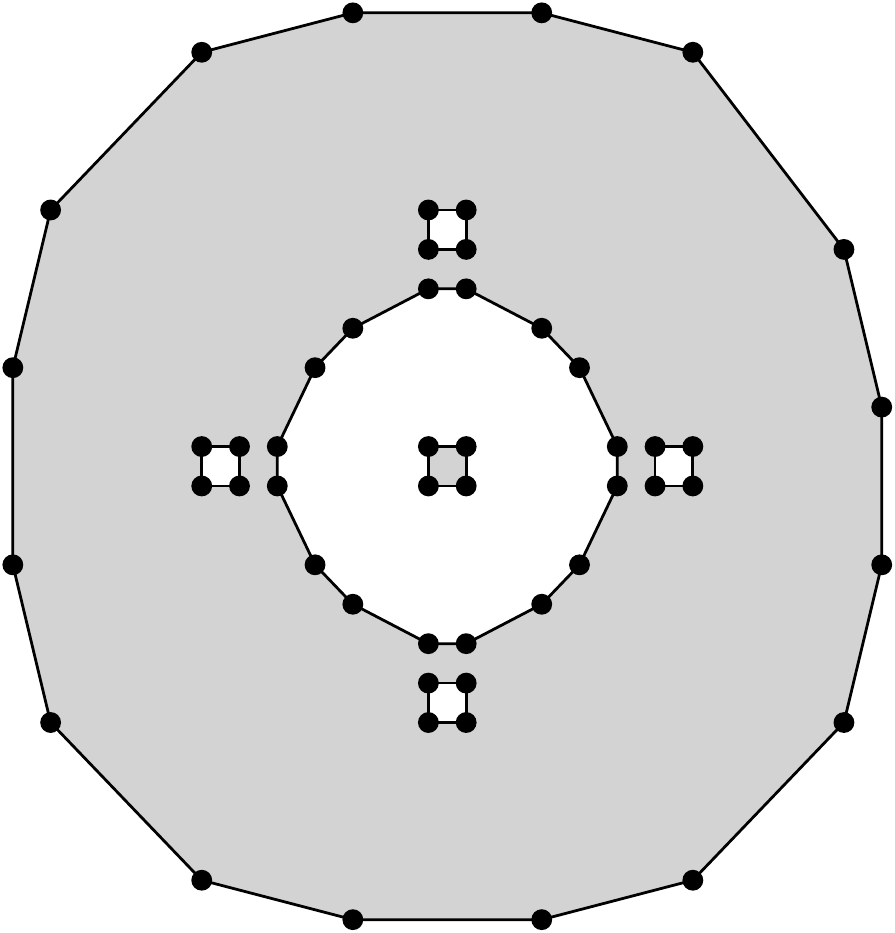}}
	\hfil
	\subfigure[]
	{\includegraphics[width=0.10\columnwidth]{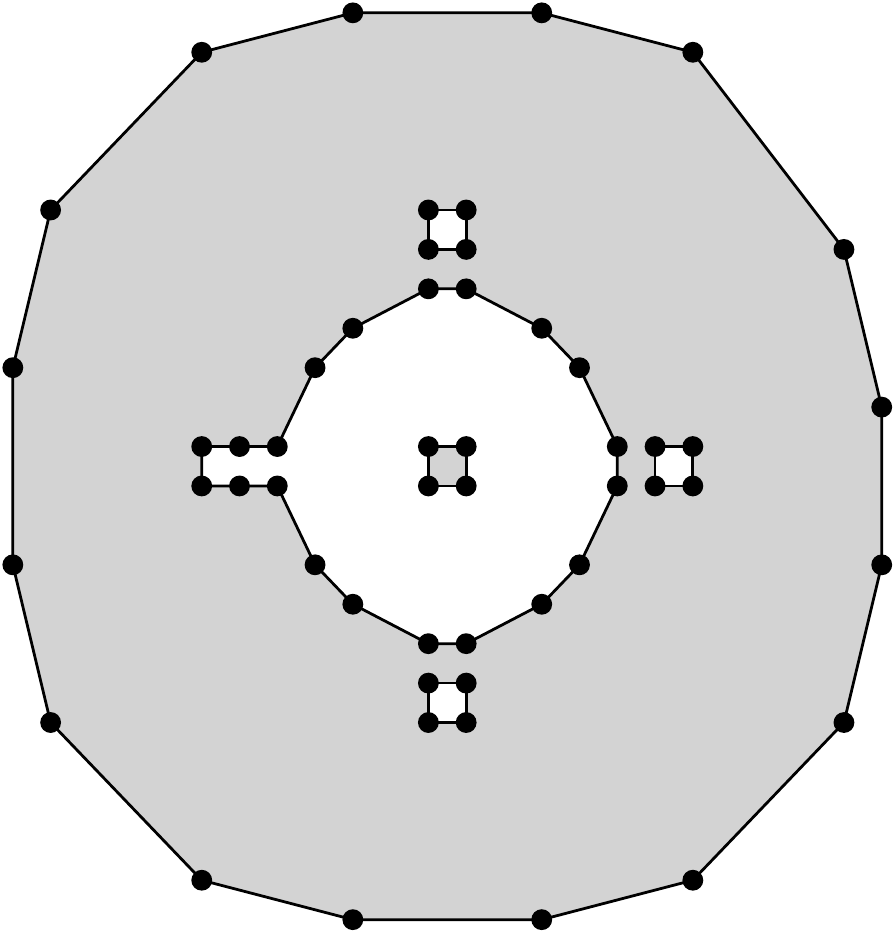}}
	\hfil
	\subfigure[]
	{\includegraphics[width=0.10\columnwidth]{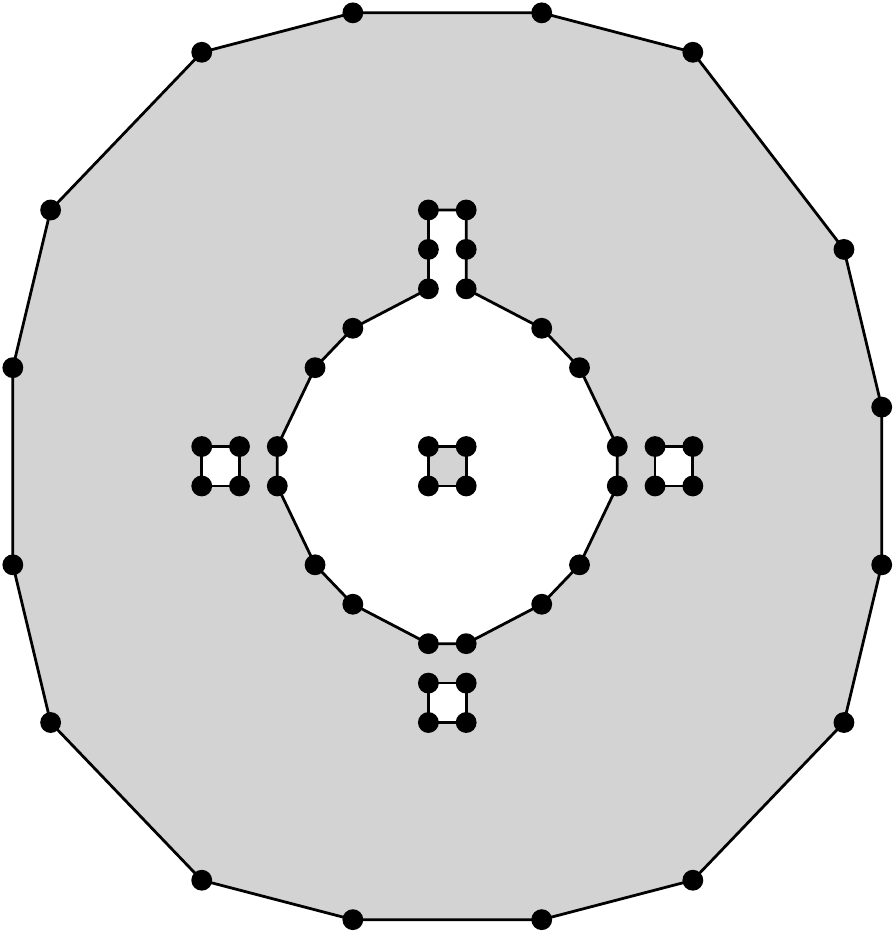}}
	\hfil
	\subfigure[]
	{\includegraphics[width=0.10\columnwidth]{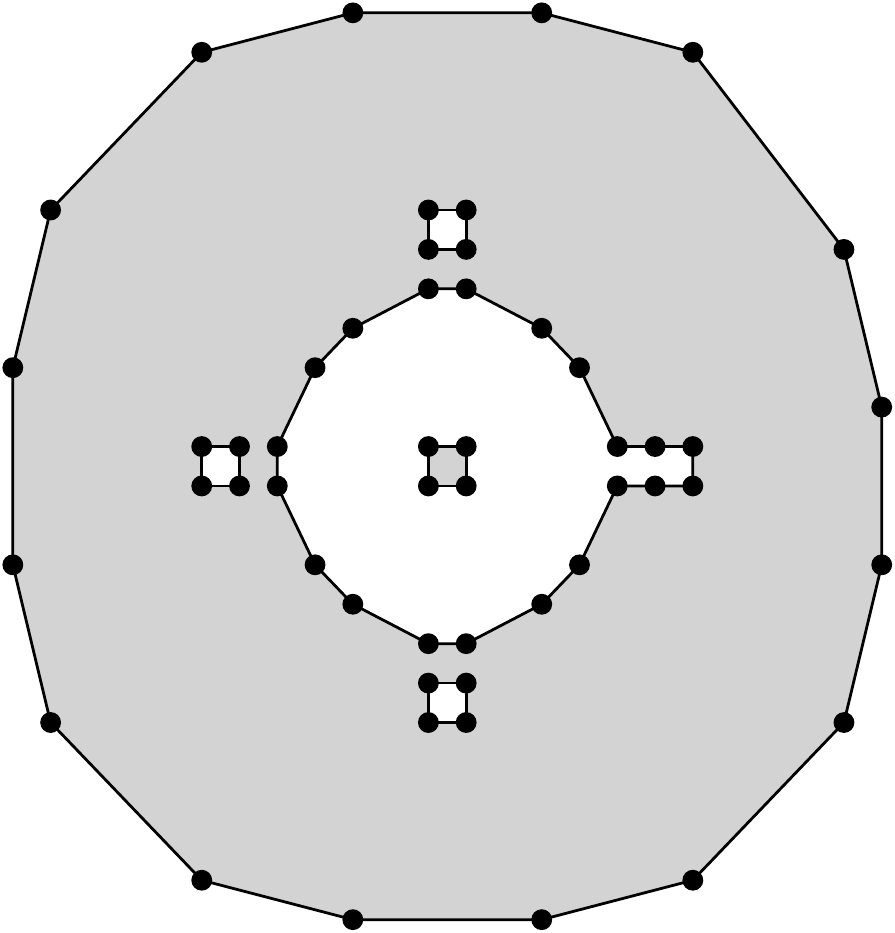}}
	
	\subfigure[]
	{\includegraphics[width=0.10\columnwidth]{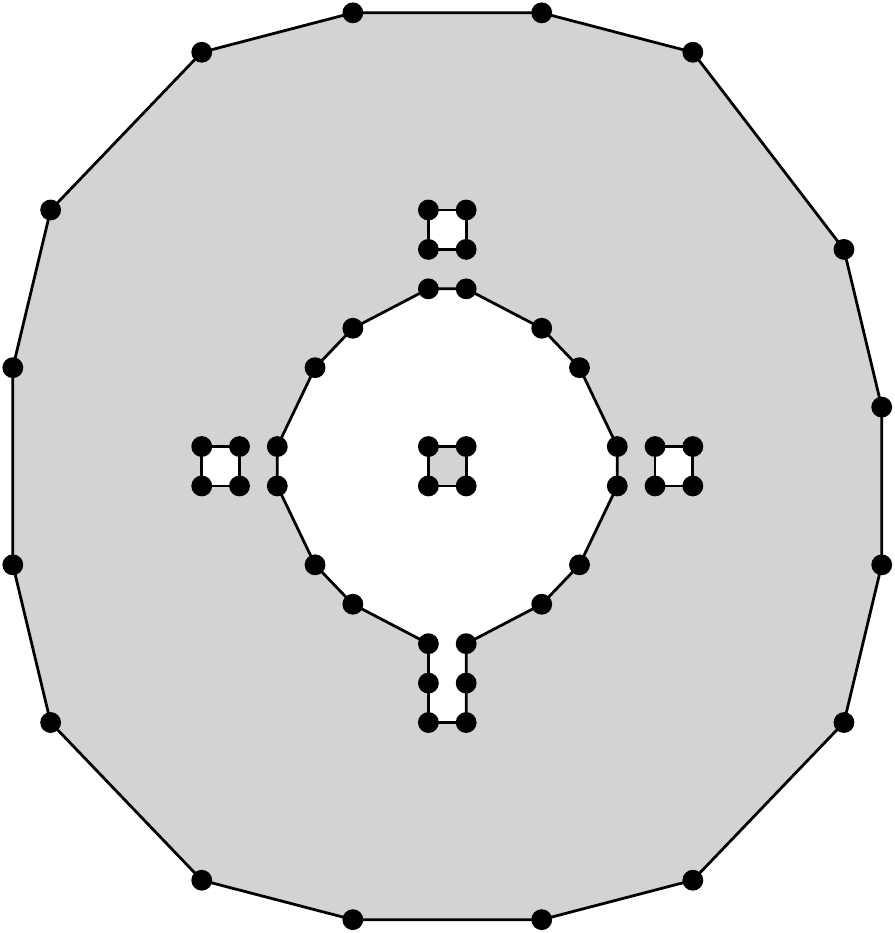}}
	\hfil
	\subfigure[]
	{\includegraphics[width=0.10\columnwidth]{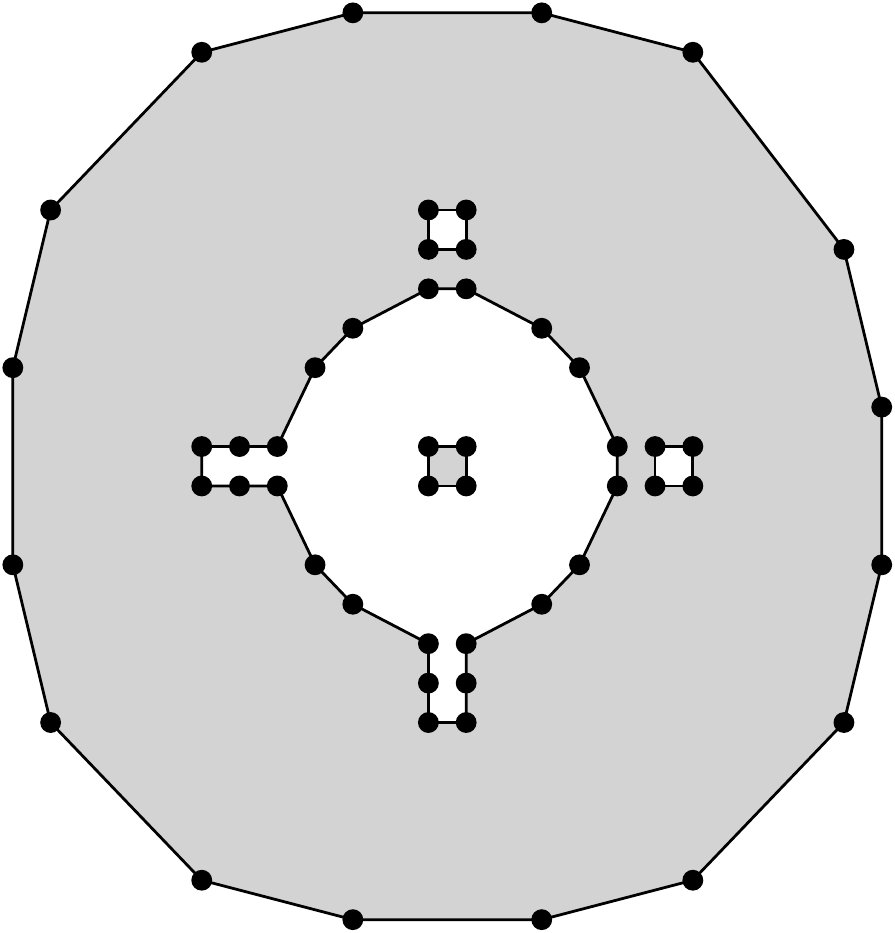}}
	\hfil
	\subfigure[]
	{\includegraphics[width=0.10\columnwidth]{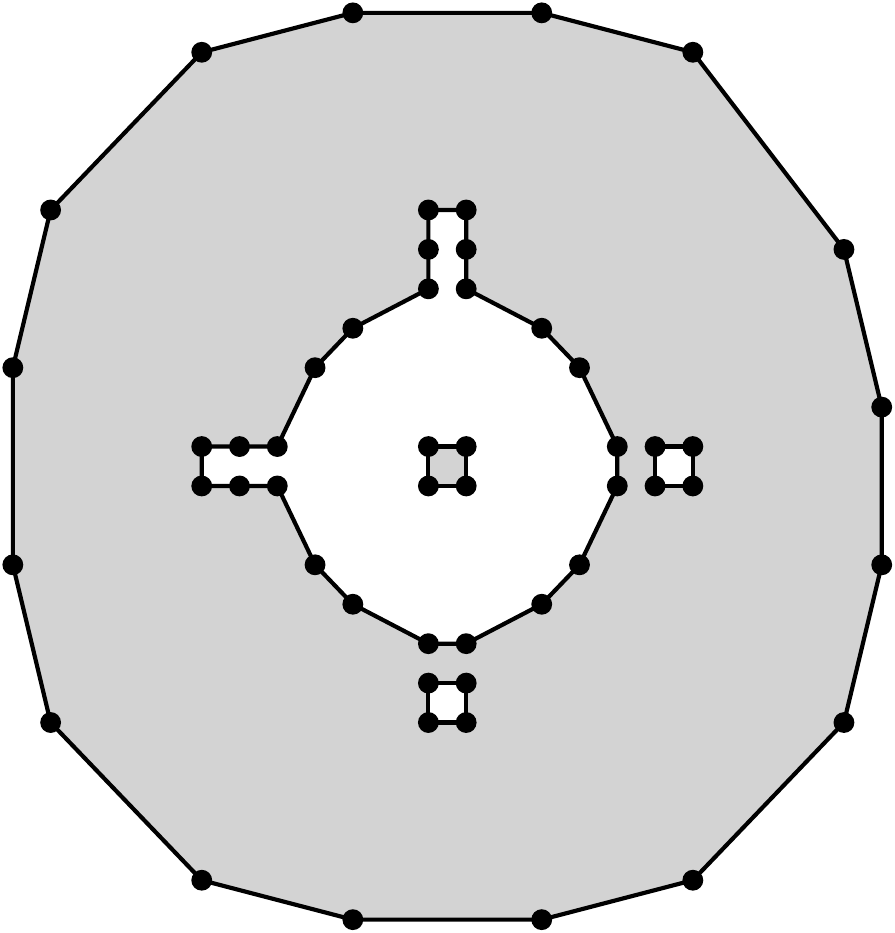}}
	
	\caption{(a) - (g) show consecutive iterations trying to solve an instance using  only constraint \eqref{invalidCircle}.}
	\label{fig:HIHCwith/out}
\end{figure}

\old{
\begin{figure}
	\centering
	\subfigure[]
	{\includegraphics[width=0.22\columnwidth]{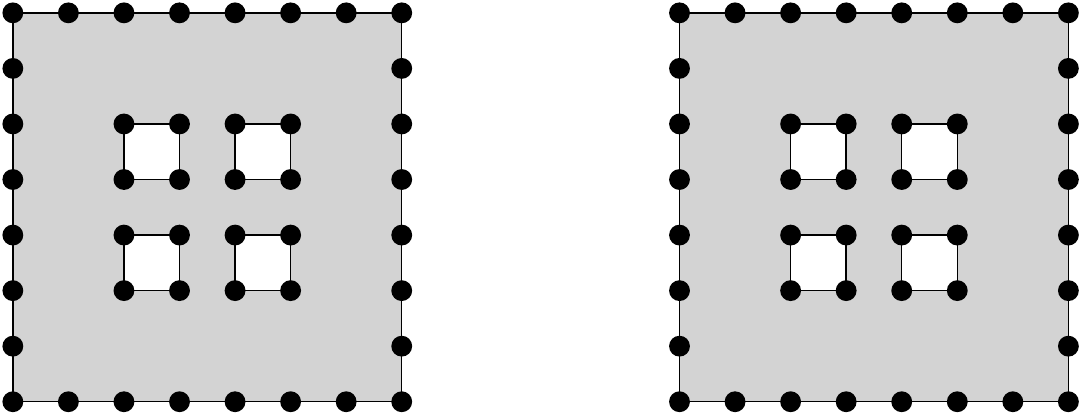}}	
	\hfil
	\subfigure[]
	{\includegraphics[width=0.22\columnwidth]{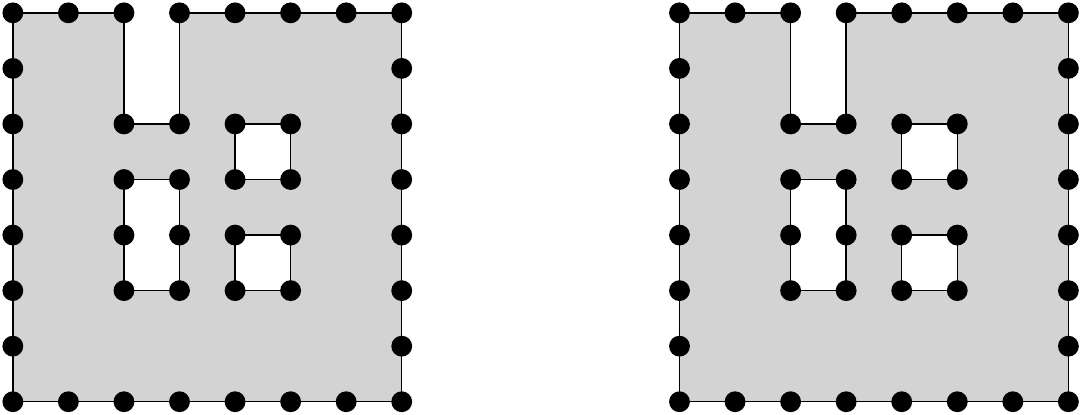}}
	\hfil
	\subfigure[]
	{\includegraphics[width=0.22\columnwidth]{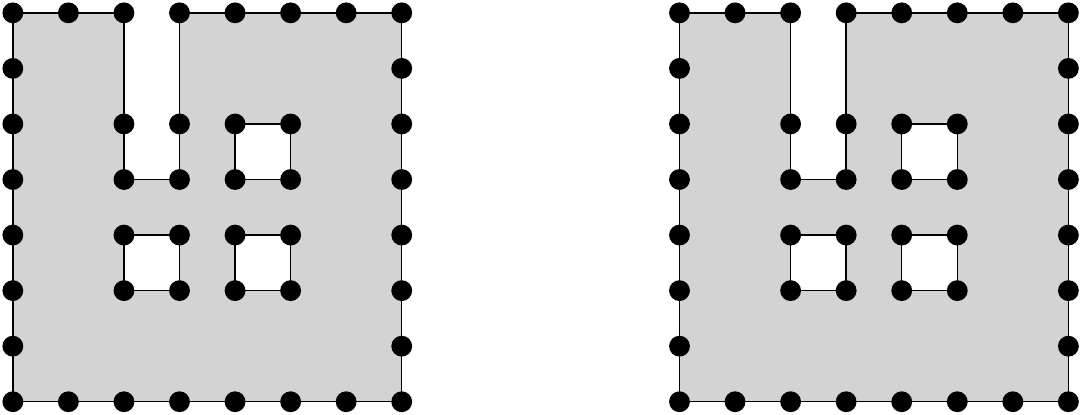}}

	\subfigure[]
	{\includegraphics[width=0.22\columnwidth]{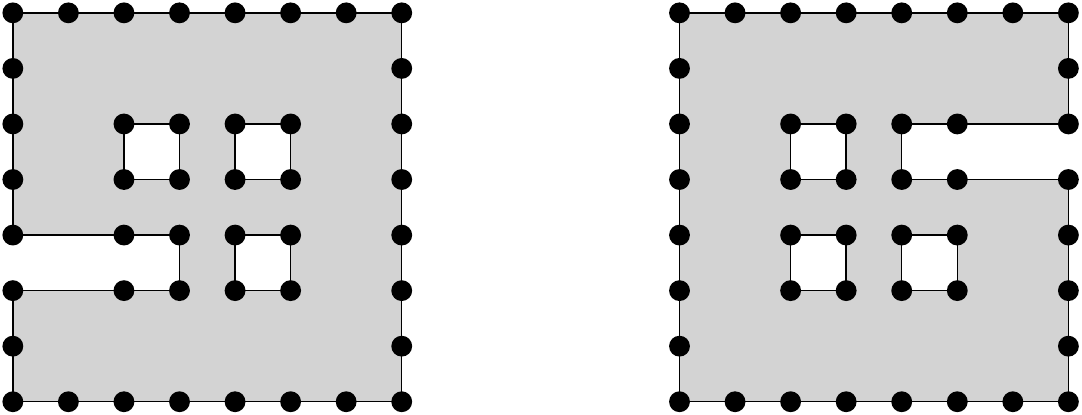}}
	\hfil
	\subfigure[]
	{\includegraphics[width=0.22\columnwidth]{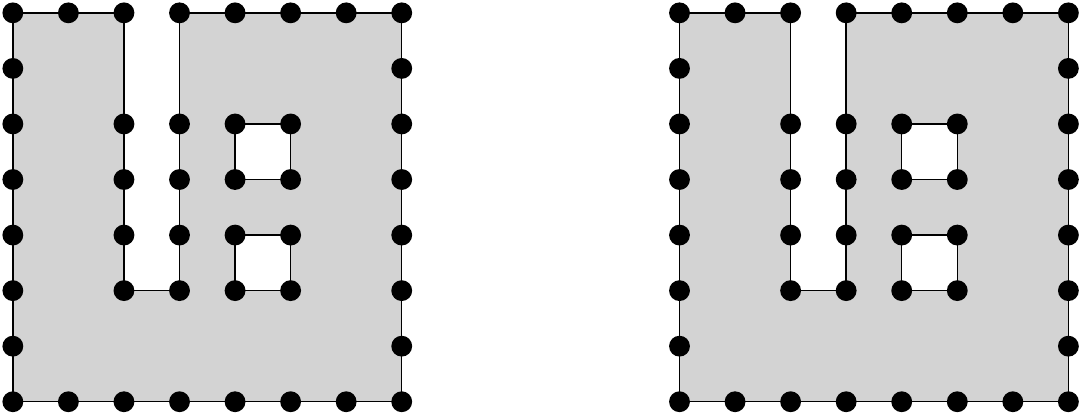}}
	\hfil	
	\subfigure[]
	{\includegraphics[width=0.22\columnwidth]{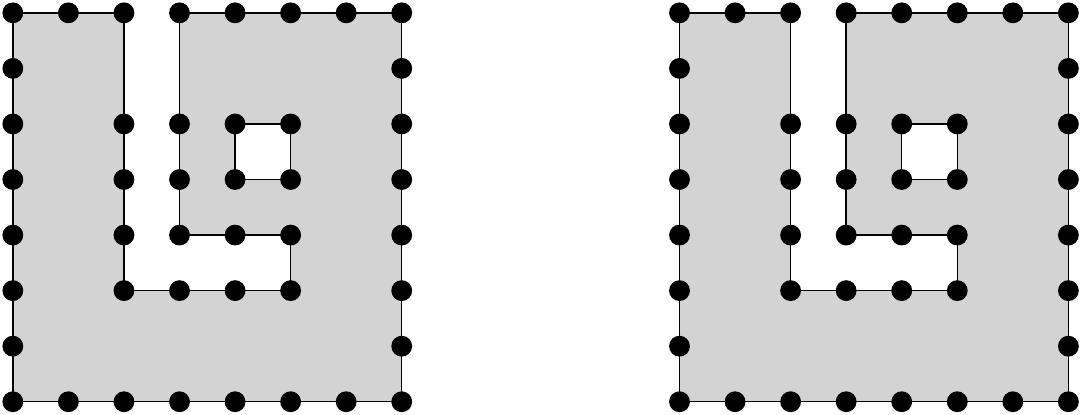}}
	\caption{(a) - (f) show consecutive iterations trying to solve an instance only using constraint \eqref{invalidCircle2}.}
	\label{fig:GCwith/out}
\end{figure}

\begin{figure}
	\centering
	\subfigure[]
	{\includegraphics[width=0.1\columnwidth,angle=90]{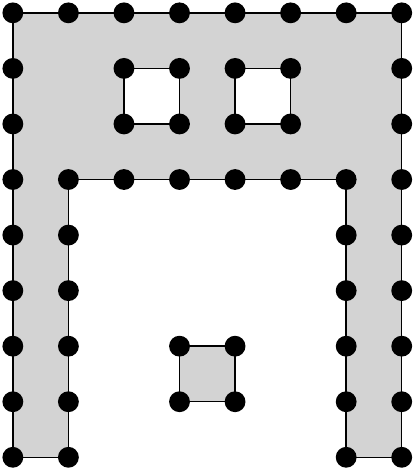}}
	\hfil
	\subfigure[]
	{\includegraphics[width=0.1\columnwidth,angle=90]{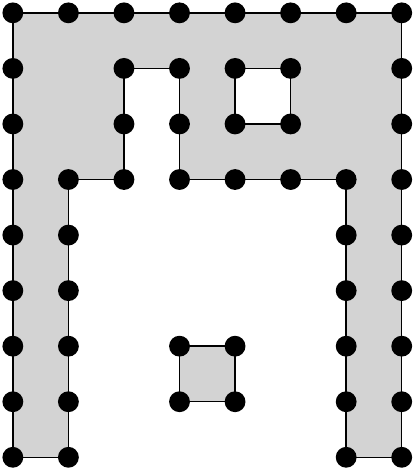}}
	\hfil
	\subfigure[]
	{\includegraphics[width=0.1\columnwidth,angle=90]{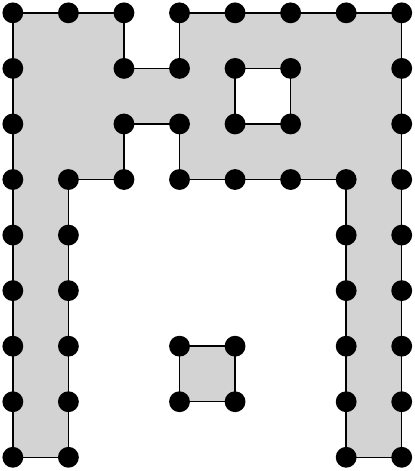}}
	\hfil
	\subfigure[]
	{\includegraphics[width=0.1\columnwidth,angle=90]{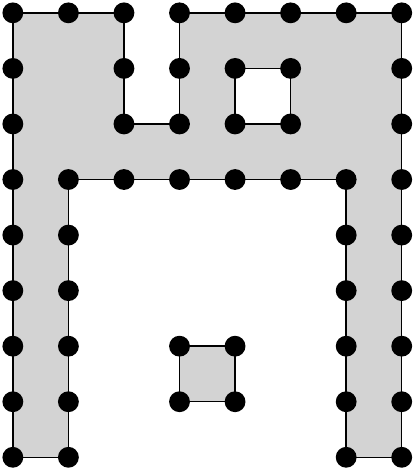}}
	
	\subfigure[]
	{\includegraphics[width=0.1\columnwidth,angle=90]{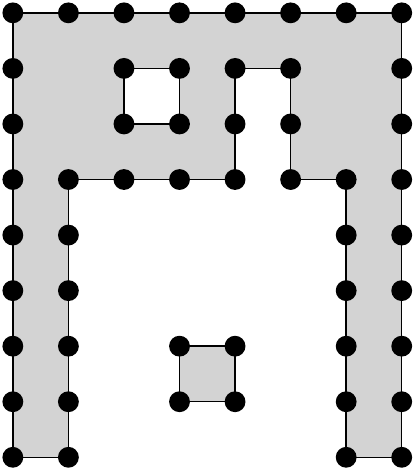}}
	\hfil
	\subfigure[]
	{\includegraphics[width=0.1\columnwidth,angle=90]{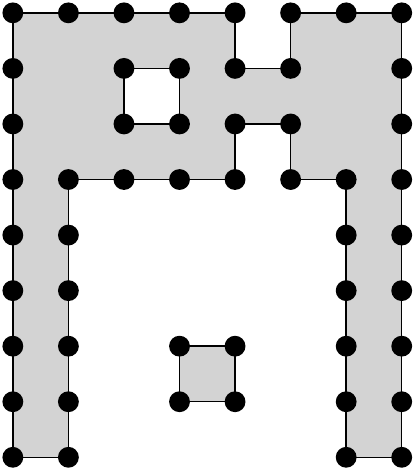}}
	\hfil
	\subfigure[]
	{\includegraphics[width=0.1\columnwidth,angle=90]{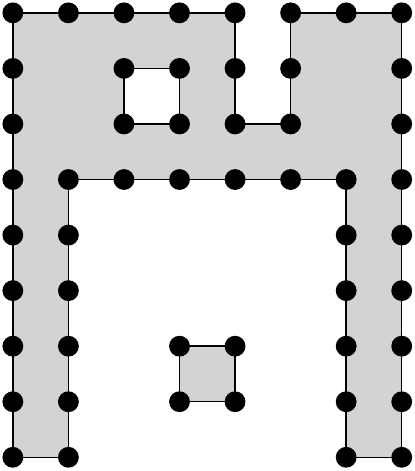}}
	\caption{(a) - (g) show consecutive iterations trying to solve an instance only using constraint \eqref{invalidCircle}.}
	\label{fig:TCwith/out}
\end{figure}

\begin{figure}
	\centering
	\subfigure[]
	{\includegraphics[width=0.15\columnwidth]{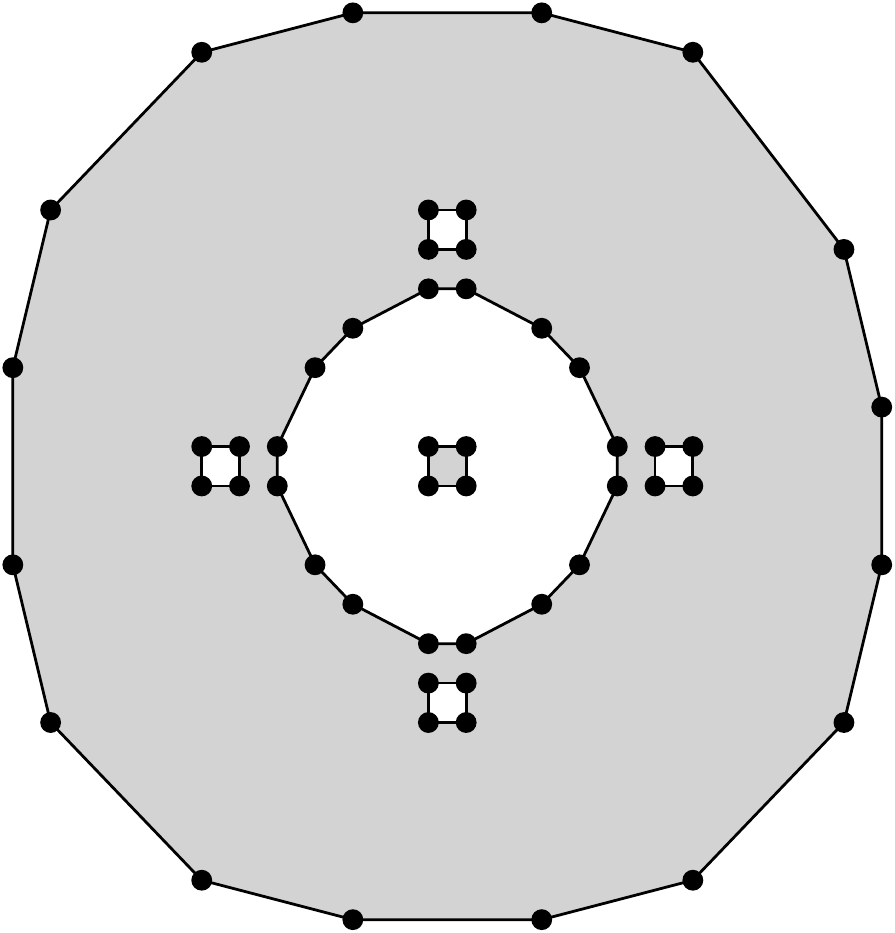}}
	\hfil
	\subfigure[]
	{\includegraphics[width=0.15\columnwidth]{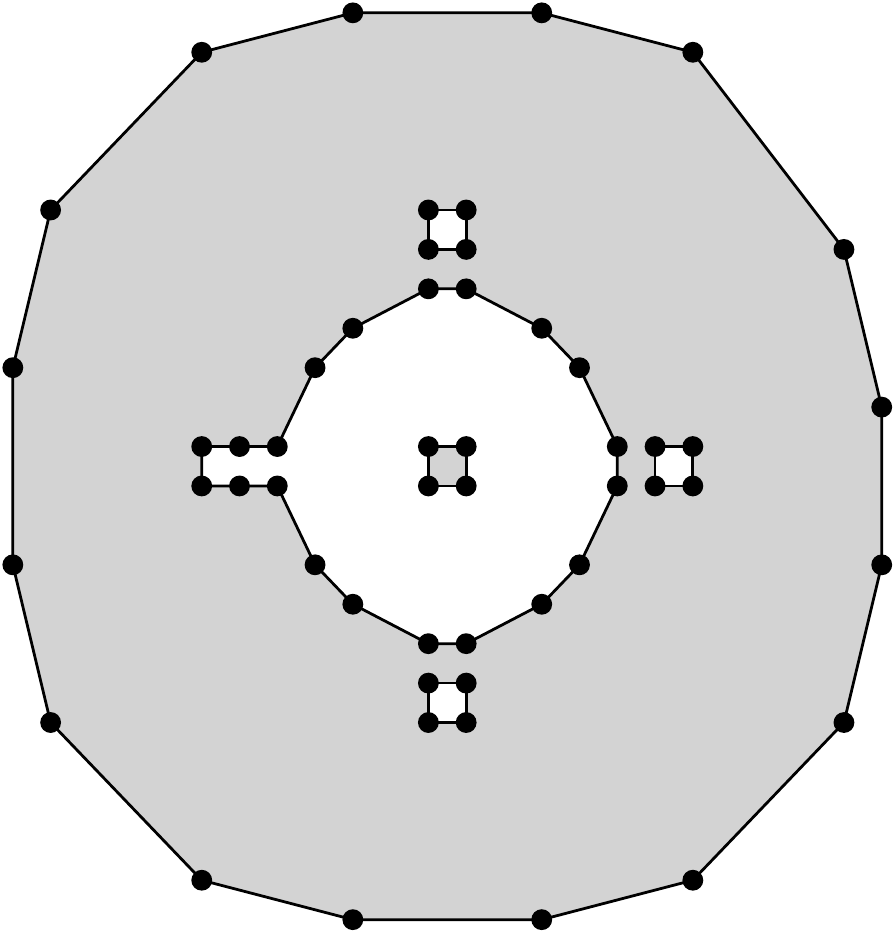}}
	\hfil
	\subfigure[]
	{\includegraphics[width=0.15\columnwidth]{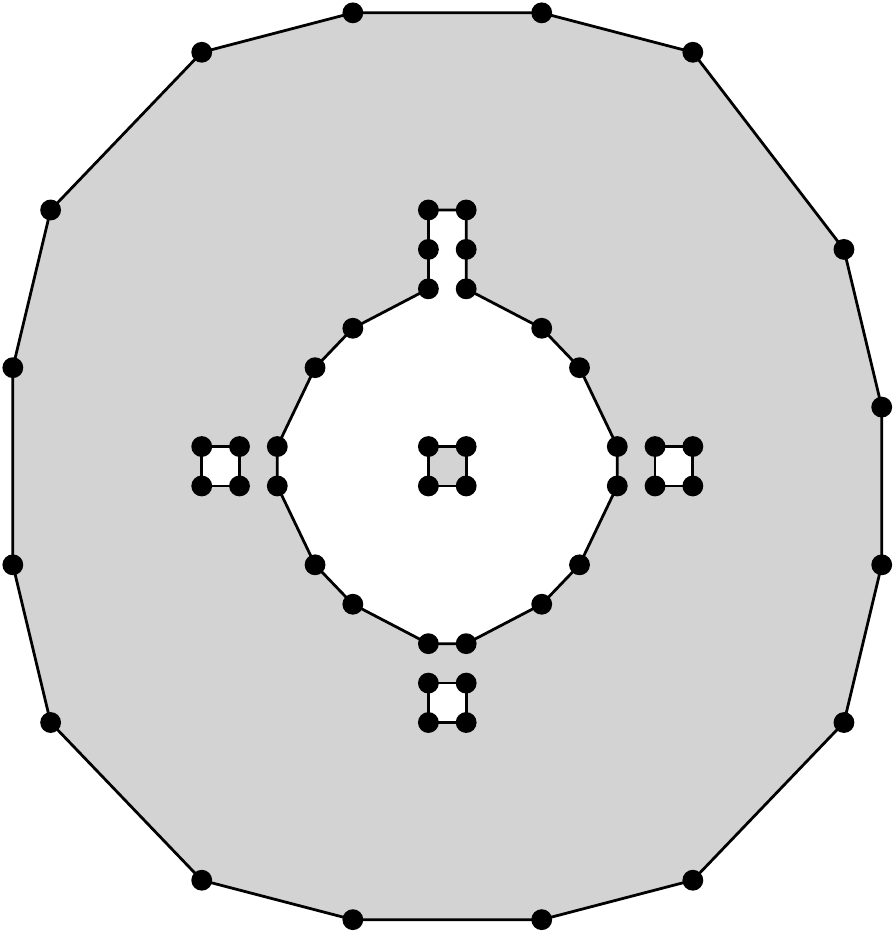}}
	\hfil
	\subfigure[]
	{\includegraphics[width=0.15\columnwidth]{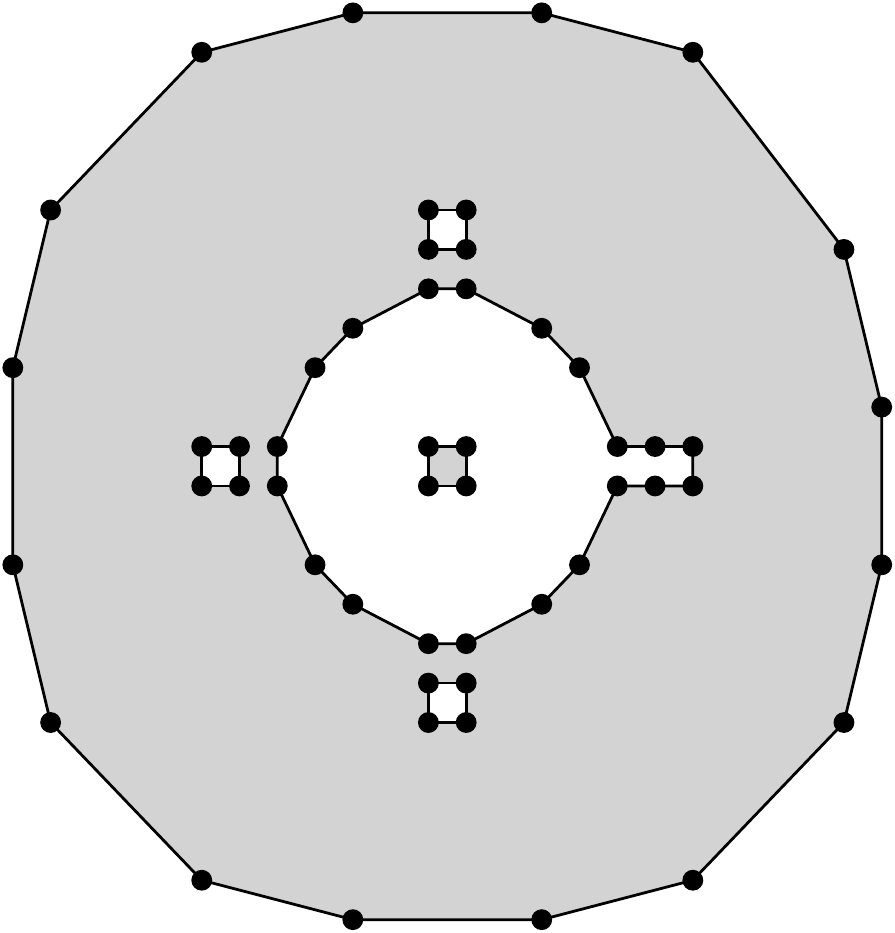}}
	
	\subfigure[]
	{\includegraphics[width=0.15\columnwidth]{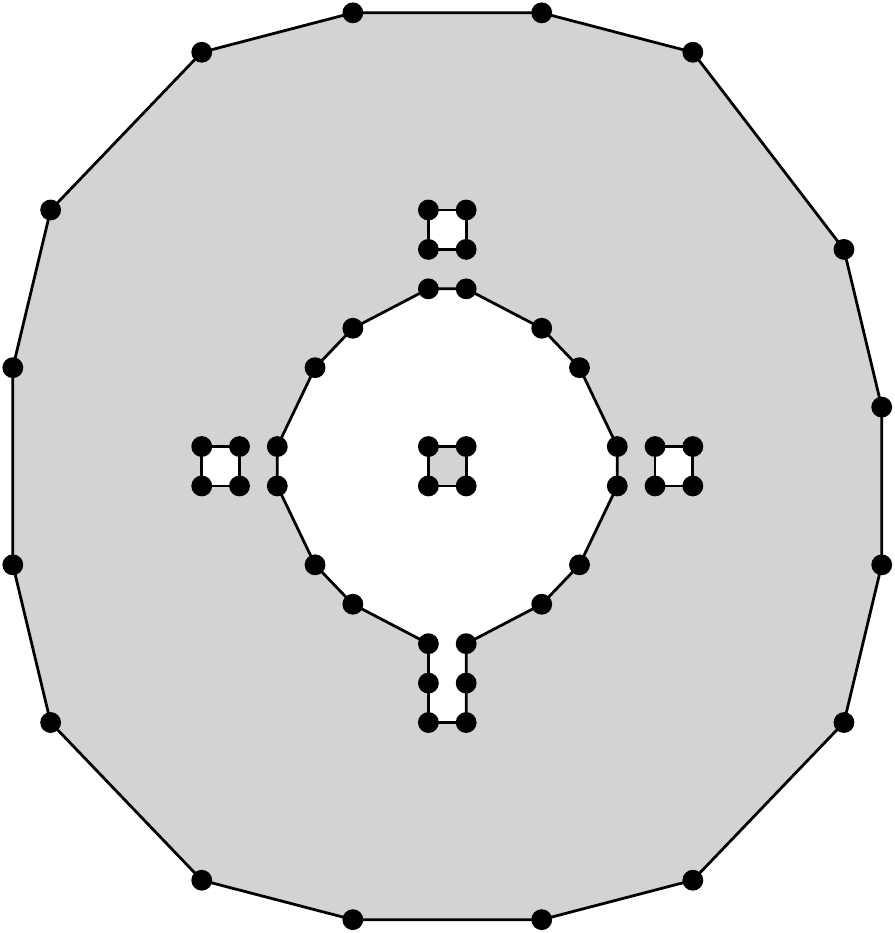}}
	\hfil
	\subfigure[]
	{\includegraphics[width=0.15\columnwidth]{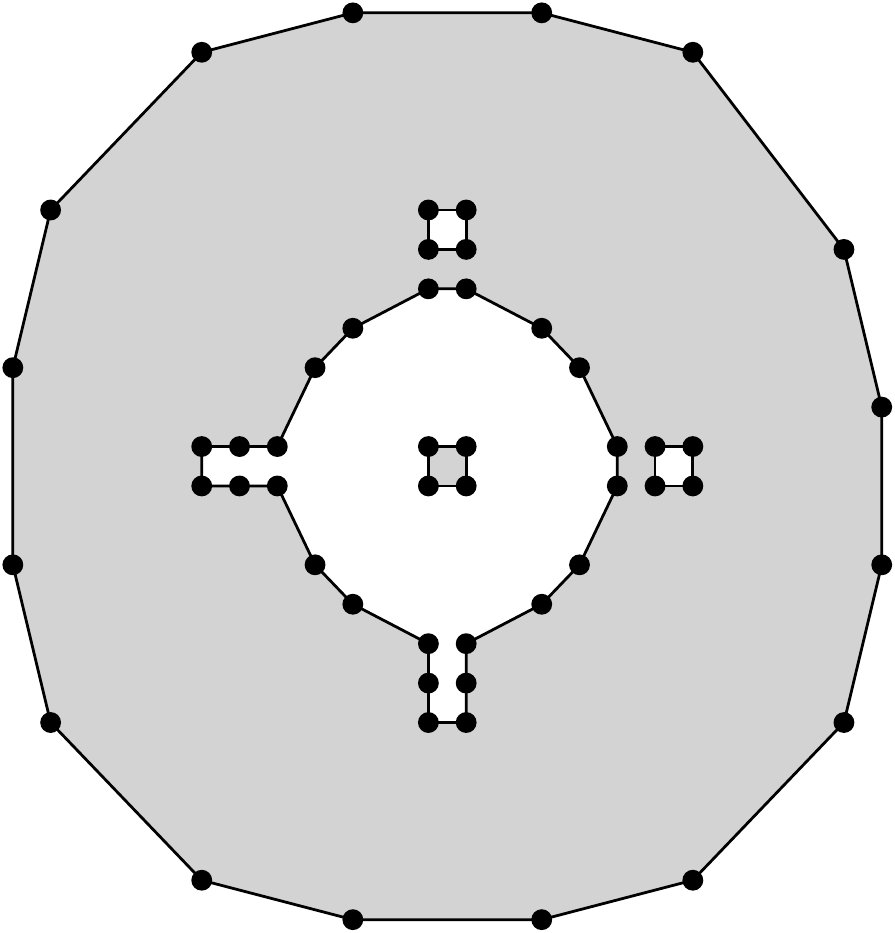}}
	\hfil
	\subfigure[]
	{\includegraphics[width=0.15\columnwidth]{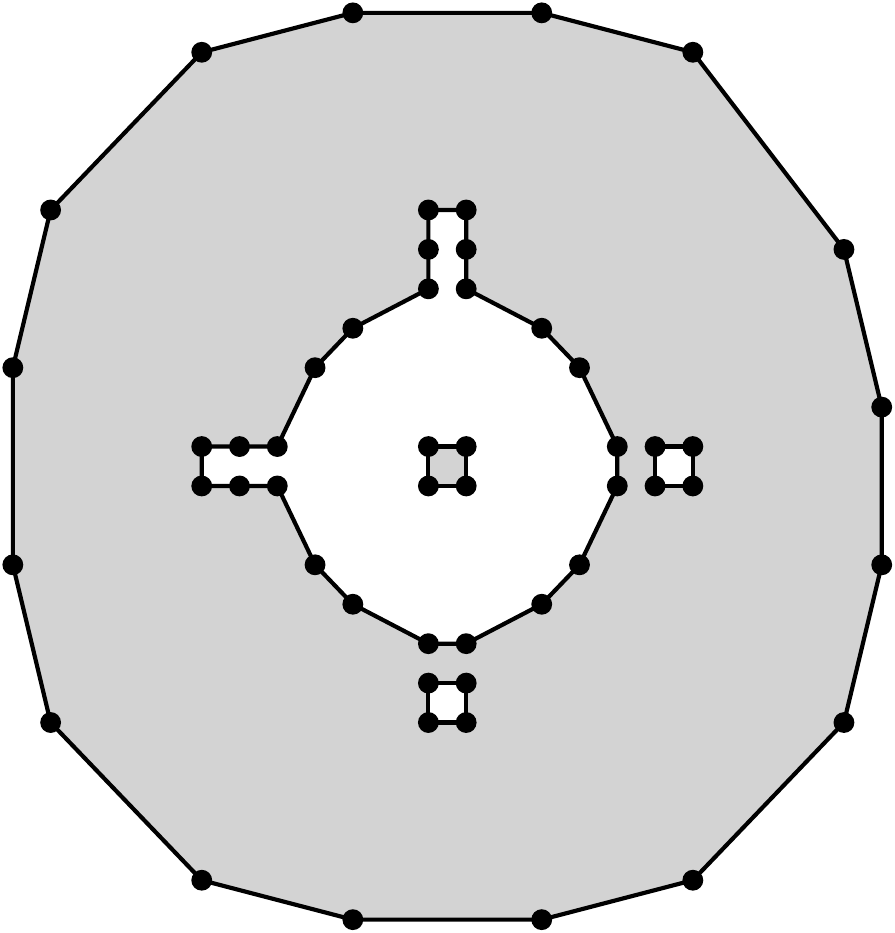}}
	
	\caption{(a) - (g) show consecutive iterations trying to solve an instance only using  constraint \eqref{invalidCircle}.}
	\label{fig:HIHCwith/out}
\end{figure}
}

The second problem is the most important, as this problem frequently becomes critical on instances of size 100 and above.
Holes in holes rarely occur on small instances but are problematic on instances of size $>200$.
The first problem  occurs only in a few instances.

In the following we describe three cuts that each solve one of the problems: The glue cut for the first problem in Section~\ref{sec:glueCut}, the tail cut for the second problem in Section~\ref{sec:tailCut}, and the HiH-Cut for the third problem in Section~\ref{sec:hihCut}.

\subsection{Glue Cuts}
\label{sec:glueCut}
To separate invalid cycles of property~\ref{prop:gluecut} we use \textit{glue cuts} (GC),
based on a curve $R_D$ from one unused convex hull edge to another (see Figure~\ref{fig:GlueCutExample}).
With $\mathcal{X}(R_D)$ denoting the set of edges crossing $R_D$,
we can add the following constraint:
\[\sum_{e\in \mathcal{X}(R_D)} x_e \geq 2\,.\]
\begin{figure}[tb]
	\centering
	\subfigure[]
	{\includegraphics[width=0.3\columnwidth]{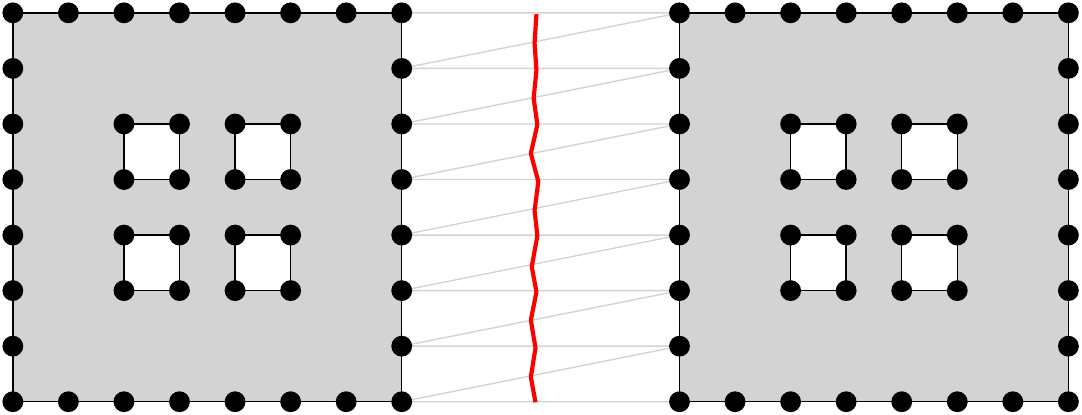}}
	\hfil
	\subfigure[]
	{\includegraphics[width=0.3\columnwidth]{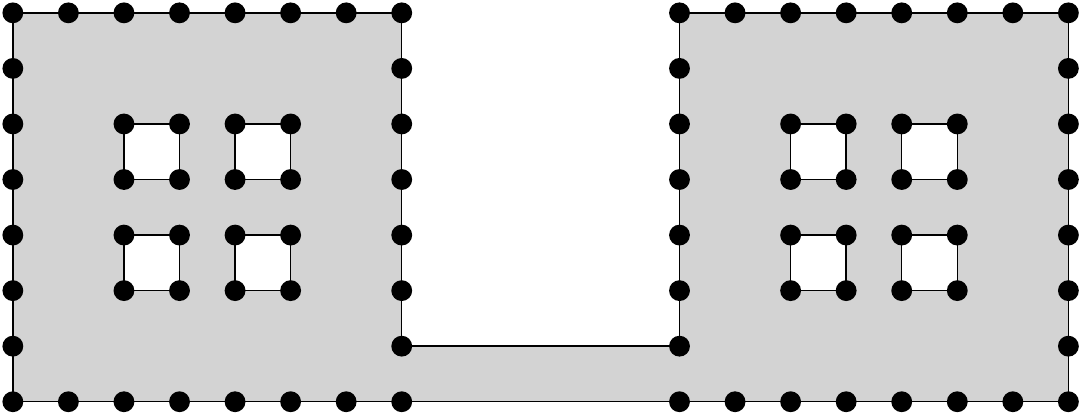}}
	\caption{Solving instance from Figure~\ref{fig:GCwith/out} with a glue cut (red). (a) The red curve needs to be crossed at least twice; it is found using the Delaunay Triangulation (grey). (b) The first iteration after using the glue cut.}
	\label{fig:GlueCutExample}
\end{figure}

Such curves can be found by considering a constrained Delaunay triangulation 
\cite{chew1989constrained} of the current solution, 
performing a breadth-first-search starting from all unused convex
hull edges of the triangulation.
Two edges are adjacent if they share a triangle.
Used edges are excluded, so our curve will not cross any used edge.
As soon as two different search trees meet, 
we obtain a valid curve by using the middle points of the edges (see red curve in Figure~\ref{fig:GlueCutExample}).

For an example, see Figure~\ref{fig:GlueCutExample}; as illustrated in Figure~\ref{fig:GCwith/out}, this instance is problematic in the Basic IP.
This can we now be solved in one iteration.

\subsection{Tail Cuts}
\label{sec:tailCut}
An outer cycle $C$ that does not contain any convex hull points cannot be simply excluded, as it may become a legal hole later.
Such a cycle either has to be merged with others, or become a hole.
For a hole,  each curve from the hole to a point outside of the convex hull must be crossed at least once.

With this knowledge we can provide the following constraint, making
use of a special curve, which we call a {\em tail} (see the red path in Figure~\ref{fig:tailcut}).

Let $R_T$ be a valid tail and $\mathcal{X}(R_T)$ the edges crossing it.
We can express the constraint in the following form:
\[\underbrace{\sum_{e\in \mathcal{X}(R_T)\setminus\delta(C)} x_e}_\text{C gets surrounded} + \underbrace{\sum_{e\in \delta(C)} x_e}_\text{C dissolves} \geq 1\,.\]

The tail is obtained in a similar fashion as the curves of the \emph{Glue Cuts}
by building a constrained Delaunay triangulation and doing a breadth-first
search starting at the edges of the cycle.  The starting points are not
considered as part of the curve and thus the curve does not cross any edges of
the current solution.

For an example, see Figure~\ref{fig:tailcut}; as illustrated in
Figure~\ref{fig:TCwith/out}, this instance is problematic in the Basic IP.
This can we now be solved in one iteration.  Note that even though it is
possible to cross the tail without making the cycle a hole, this is more
expensive than simply merging it with other cycles.

\ignore{While solving the IP it can happen that we get an intermediate solution where a cycle $C$ with no convex hull point exists that is not enclosed by another component. This is equivalent to the property \ref{prop:tailcut} of invalid cycles.}
\ignore{
If we have invalid cycles of property \ref{prop:tailcut}, we know that there exists a cycle $C$ with no convex hull points which is not enclosed by another cycle.  
In a separation step we cannot state if $C$ will be a hole of the optimal solution or if it will be connected to the outer boundary in some way. That means one case of the following will be true in an optimal solution:
\begin{enumerate}
	\item $C$ is surrounded by edges of the outer boundary.
	\item At least two edges leave $C$.
\end{enumerate}

We will now construct a constraint that covers both cases at the same time. Consider a path $R_T$ from $C$ to an unused convex hull edge (see Figure~\ref{fig:tailcut}). Since $C$ is not enclosed by any other component this path must exist. Now let $\mathcal{X}(R_T)$ be the set of edges crossing $R_T$.
We sum up over all edges crossing $R_T$ which do not start in $C$ plus all edges leaving $C$.
Hence we gain the following constraint:
\[\underbrace{\sum_{e\in \mathcal{X}(R_T)\setminus\delta(C)} x_e}_\text{C gets surrounded} + \underbrace{\sum_{e\in \delta(C)} x_e}_\text{C dissolves} \geq 1\]
We call this constraint a \textit{tail cut} (TC). Note that we may have to do more than one cut to ensure that one case will be true.}
\begin{figure}[tb]
	\centering
		\subfigure[]
		{\includegraphics[width=0.15\columnwidth,angle=90]{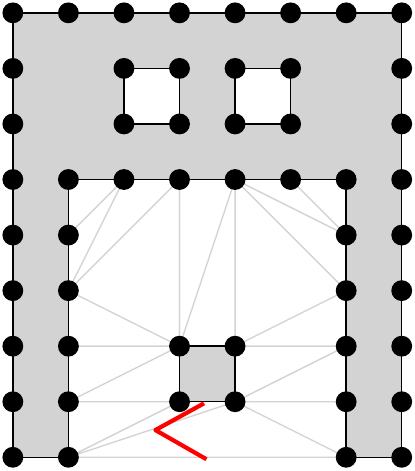}}
		\hfil
		\subfigure[]
		{\includegraphics[width=0.15\columnwidth,angle=90]{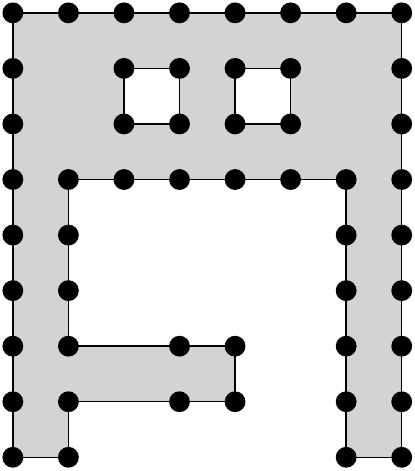}}
		\caption{Solving the instance from Figure~\ref{fig:TCwith/out} with tail cut (red line). (a) The red curve needs to be crossed at least twice or two edges must leave the component. The red curve is found via the Delaunay Triangulation (grey). (b) The first Iteration after using the tail cut.}
		\label{fig:tailcut}
\end{figure}

\subsection{Hole-in-Hole Cuts}
\label{sec:hihCut}

The difficulty of eliminating holes in holes (Problem 3) is
that they may end up as perfectly legal simple holes, if the outer cycle gets merged with the outer boundary.
In that case, every curve from the hole to the convex hull \emph{cannot} cross the used edges exactly two times (edges of the hole are ignored).
One of the crossed edges has to be of the exterior cycle and the other one cannot be as otherwise we would be outside again.
It can also be not of an interior cycle, as we would have to leave it again to reach our hole.

Therefore the inner cycle of a hole in hole either has to be merged, or all curves from it to the convex hull do not have exactly two used edge crossings.
As it is impractical to argue over all curves, we only pick one curve $P$ that currently crosses exactly two used edges 
(see the red curve in Figure~\ref{fig:holeinholcut} with crossed edges in green).

Because we cannot express the inequality that $P$ is not allowed to be crossed exactly two times as an linear programming constraint,
we use the following weaker observation.
If the cycle of the hole in hole becomes a simple hole, the crossing of $P$ has to change.
Let $e_1$ and $e_2$ be the two used edges that currently cross $P$ and $\mathcal{X}(P)$ the set of all edges crossing $P$ (including unused but no edges of $H$).
We can express a change on $P$ by
\[\underbrace{\sum_{e\in \mathcal{X}(P)\setminus\{e_1,e_2\}}x_e}_{\text{new crossing}} + \underbrace{-x_{e_1}-x_{e_2}}_{\text{$e_1$ or $e_2$ vanishes}} \geq -1\,.\]
%
Together we obtain the following LP-constraint for either $H$ being merged or the crossing of $P$ changes:
\[\underbrace{\sum_{e\in \delta(V_H, V\setminus V_H)}x_e}_{\text{$H$ dissolves}} + \underbrace{\sum_{e\in \mathcal{X}(P)\setminus\{e_1,e_2\}}x_e + -x_{e_1}-x_{e_2}}_{\text{Crossing of $P$ changes}} \geq -1\,.\]

Again we use a breadth-first search on the constrained Delaunay triangulation
starting from the edges of the hole in hole. Unlike the other two cuts we need
to cross used edges. Thus, we get a shortest path search such that the optimal
path primarily has a minimal number of used edges crossed and secondarily has a
minimal number of all edges crossed.

\begin{figure}
	\centering
	\subfigure[]
	{\includegraphics[width=0.2\columnwidth]{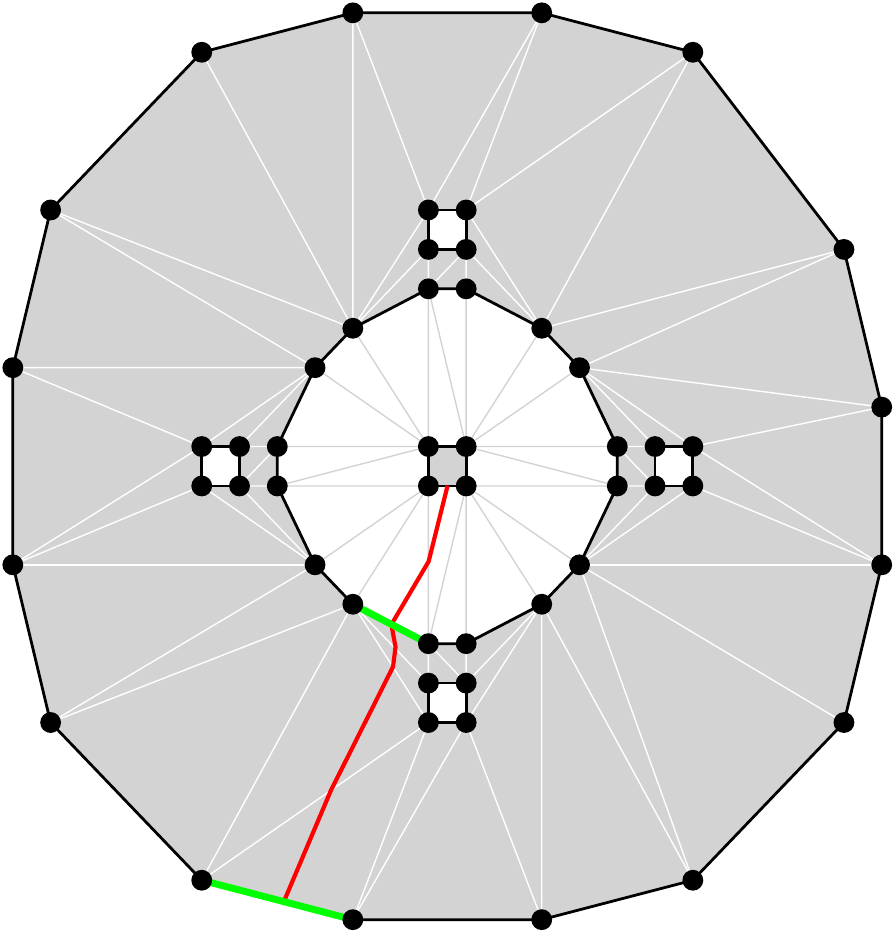}}
	\hfil
	\subfigure[]
	{\includegraphics[width=0.2\columnwidth]{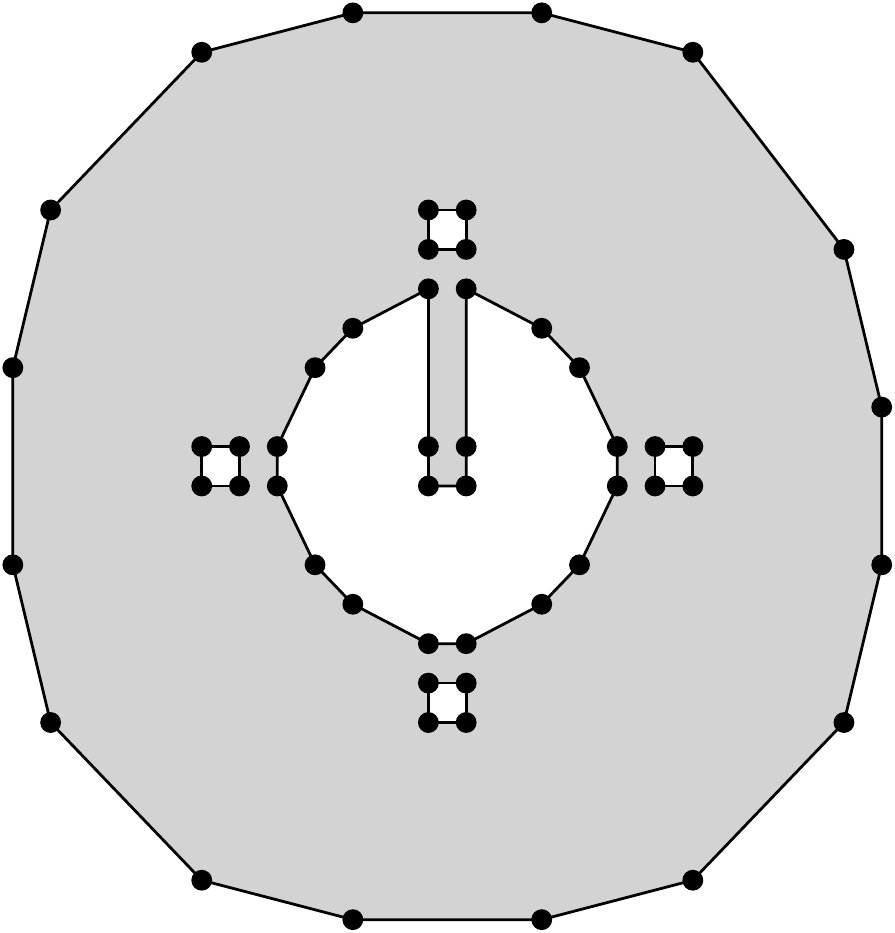}}
	\caption{Solving instance from Figure~\ref{fig:HIHCwith/out} with hole in hole cut (red line). (a) Red line needs to be crossed at least two times or two edges must leave the component or one of the two existing edges (green) must be removed. Red line is built via Delaunay Triangulation. (b) First Iteration after using the hole in hole cut.}
	\label{fig:holeinholcut}
\end{figure}

For an example, see Figure~\ref{fig:holeinholcut}; as illustrated in Figure~\ref{fig:GCwith/out}, this instance is problematic in the Basic IP.
This can we now be solved in one iteration.
The corresponding path is displayed in red and the two crossed edges are highlighted in green.
Changing the crossing of the path is more expensive than simply connect the hole in hole to the outer hole and thus the hole in hole dissolves.

\section{Experiments}
\label{sec:experiments}

\subsection{Implementation}
Our implementation uses CPLEX 
to solve the relevant IPs. 
Important is also the geometric side of computation, for which we used 
the CGAL Arrangements package~\cite{cgal:arr}.
CGAL represents a planar subdivision using a doubly connected edge list (DCEL), which is ideal for detecting invalid boundary cycles.

\subsection{Test Instances}

While the TSPLIB is well-recognized and offers a good mix of instances with different
structure (ranging from grid-like instances over relatively uniform random distribution
to highly clustered instances), it is relatively sparse. Observing that the
larger TSPLIB instances are all geographic in nature, we
designed a generic approach that yields arbitrarily large and numerous clustered instances.
This is based on illumination maps: A satellite image of a geographic region at night time
displays uneven light distribution.
 The corresponding brightness values can be used as a random density function that can be 
used for sampling (see Figure \ref{fig::generator}). To reduce noise, we cut off brightness values
below a certain threshold, i.e., we set the probability of choosing the respective 
pixels to zero.

\old{
\begin{figure}
        \subfigure[Earth by night \label{fig:earthByNight}] {\includegraphics[width=0.45\columnwidth]{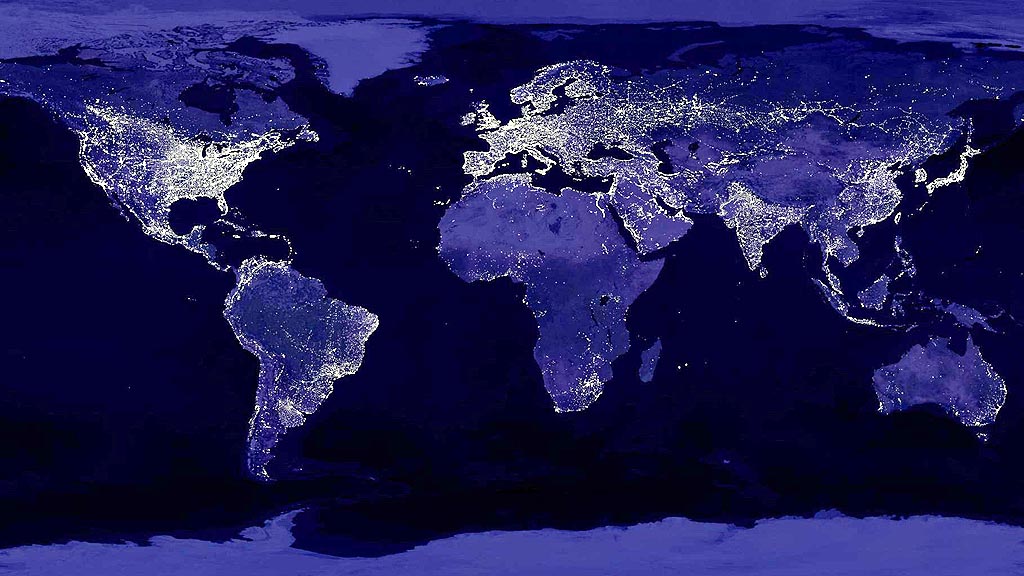}}
        \hfill
        \subfigure[A sampled instance]{\includegraphics[width=0.45\columnwidth]{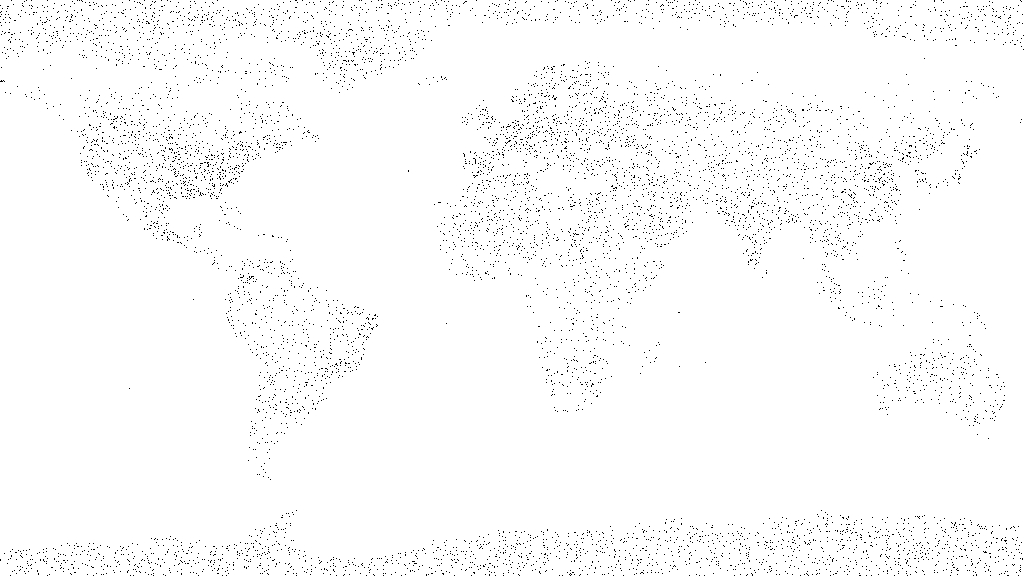}}
        \caption{Using the brightness of an image as a density function to generate clustered point sets.}
        \label{fig::generator}
\end{figure}
}

%

\subsection{Results}
All experiments were run on an \textit{Intel Core i7-4770} CPU clocked at 3.40 GHz with 16 GB of RAM.
We set a 30 minute time limit to solve the instances. In Table \ref{tab:TSPLib}, all
results are displayed for every instance 
that we solved within the time limit. The
largest instance solved within 30 minutes is gr666 with 666 points, which took about 6 minutes. The
largest instance solved out of the TSPLib so far is dsj1000 with 1000 points, solved in about 37 minutes.
In addition, we generated 30 instances for each size, which were run with a time limit of 30 minutes. 

    \begin{longtable}{rcccccc}
    	    	    \caption{The runtime in milliseconds of all variants on the instances of the TSPLib that we solved within 30 minutes. The number in the name of an instance indicates the number of points.}\\
    \toprule
    & BasicIP & +JS+DC & +JS+TC & +JS+DC & +JS+DC & +DC+TC \\
    &         & +TC+HIHC & +HIHC & +HIHC & +TC    & +HIHC \\
    \midrule
    \endfirsthead
    \multicolumn{7}{c}%
    {\tablename\ \thetable\ -- \textit{Continued from previous page}} \\
    \toprule
    & BasicIP & +JS+DC & +JS+TC & +JS+DC & +JS+DC & +DC+TC \\
    &         & +TC+HIHC & +HIHC & +HIHC & +TC    & +HIHC \\
    \midrule
    \endhead
    \hline \multicolumn{7}{r}{\textit{Continued on next page}} \\
    \endfoot
    \endlastfoot
    burma14 & 20    & 22    & 17    & 19    & 26    & 19 \\
    ulysses16 & 48    & 42    & 35    & 43    & 32    & 42 \\
    ulysses22 & 50    & 34    & 55    & 31    & 32    & 61 \\
    att48 & 180   & 58    & 72    & 62    & 57    & 129 \\
    eil51 & 74    & 82    & 72    & 78    & 81    & 99 \\
    berlin52 & 43    & 38    & 37    & 37    & 38    & 51 \\
    st70  & -     & 329   & 324   & -     & 348   & 414 \\
    eil76 & 714   & 144   & 105   & 530   & 148   & 239 \\
    pr76  & -     & 711   & 711   & -     & 731   & 1238 \\
    gr96  & 376   & 388   & 349   & 10982 & 384   & 367 \\
    rat99 & 922   & 480   & 485   & 464   & 513   & 1190 \\
    kroA100 & -     & 961   & 689   & -     & 950   & 1294 \\
    kroB100 & -     & 1470  & 2623  & -     & 1489  & 2285 \\
    kroC100 & -     & 470   & 431   & -     & 465   & 577 \\
    kroD100 & 4673  & 509   & 451   & 4334  & 514   & 835 \\
    kroE100 & -     & 273   & 273   & -     & 272   & 574 \\
    rd100 & -     & 894   & 756   & -     & 890   & 2861 \\
    eil101 & -     & 575   & 445   & -     & 527   & 1090 \\
    lin105 & -     & 390   & 359   & -     & 412   & 931 \\
    pr107 & 550   & 401   & 272   & 346   & 513   & 923 \\
    pr124 & 495   & 348   & 264   & 322   & 355   & 940 \\
    bier127 & 439   & 288   & 270   & 267   & 276   & 476 \\
    ch130 & -     & 1758  & 1802  & -     & 1594  & 2853 \\
    pr136 & 1505  & 964   & 1029  & 992   & 950   & 3001 \\
    gr137 & -     & 1262  & 1361  & -     & 1252  & 1724 \\
    pr144 & 6276  & 1028  & 2926  & 985   & 1030  & 2012 \\
    ch150 & -     & 4938  & 5167  & -     & 5867  & 7997 \\
    kroA150 & -     & 3427  & 5615  & -     & 3327  & 7474 \\
    kroB150 & -     & 2993  & 2396  & -     & 2943  & 5265 \\
    pr152 & 13285 & 2161  & 1619  & 10978 & 2151  & 19479 \\
    u159  & 13285 & 1424  & 1262  & 5339  & 1410  & 2513 \\
    rat195 & 106030 & 16188 & 19780 & 77216 & 16117 & 27580 \\
    d198  & -     & 19329 & 155550 & -     & 19398 & 41118 \\
    kroA200 & -     & 26360 & 13093 & -     & 26389 & 11844 \\
    kroB200 & -     & 5492  & 6239  & -     & 5525  & 15238 \\
    gr202 & -     & 4975  & 7512  & -     & 4304  & 9670 \\
    ts225 & 18902 & 7746  & 9750  & 7595  & 7603  & 60167 \\
    tsp225 & 91423 & 11600 & 9741  & 28756 & 11531 & 44297 \\
    pr226 & -     & 8498  & 2800  & -     & 7204  & 18848 \\
    gr229 & -     & 5462  & 26478 & -     & 10153 & 25674 \\
    gil262 & -     & 23000 & 22146 & -     & -     & 72772 \\
    pr264 & 24690 & 6537  & -     & 6719  & 6549  & 23641 \\
    a280  & 22023 & 3601  & 3857  & 3980  & 3619  & 12983 \\
    pr299 & -     & 16251 & 355323 & -     & 16173 & 85789 \\
    lin318 & -     & 23863 & 1511219 & -     & 24035 & 75312 \\
    linhp318 & -     & 23107 & 1313680 & -     & 23064 & 79352 \\
    rd400 & -     & 111128 & 92995 & -     &       & 302363 \\
    fl417 & -     & 198013 & -     & -     & 215210 & 825808 \\
    gr431 & -     & 56716 & 173609 & -     & 78133 & 265416 \\
    pr439 & -     & 46685 & 36592 & -     & 48231 & 273873 \\
    pcb442 & -     & 1356796 & -     & -     & -     & - \\
    d493  & -     & 359072 & -     & -     & -     & 837229 \\
    att532 & -     & 217679 & 256394 & -     & 218665 & 817096 \\
    ali535 & -     & 93771 & 427800 & -     & 91828 & 323104 \\
    u574  & -     & 371523 & 199114 & -     & -     & 1010276 \\
    rat575 & -     & 417494 & 191198 & -     & 580320 & 934988 \\
    p654  & -     & 864066 & -     & -     & -     & - \\
    d657  & -     & 455378 & 253374 & -     & 646148 & 1352747 \\
    gr666 & -     & 366157 & -     & -     & 670818 & - \\
    \bottomrule
    \label{tab:TSPLib}%
    \end{longtable}%

\begin{figure}
	\begin{center}
		\includegraphics[width=0.48\textwidth, trim= 0mm 40mm 0mm 0mm, clip]{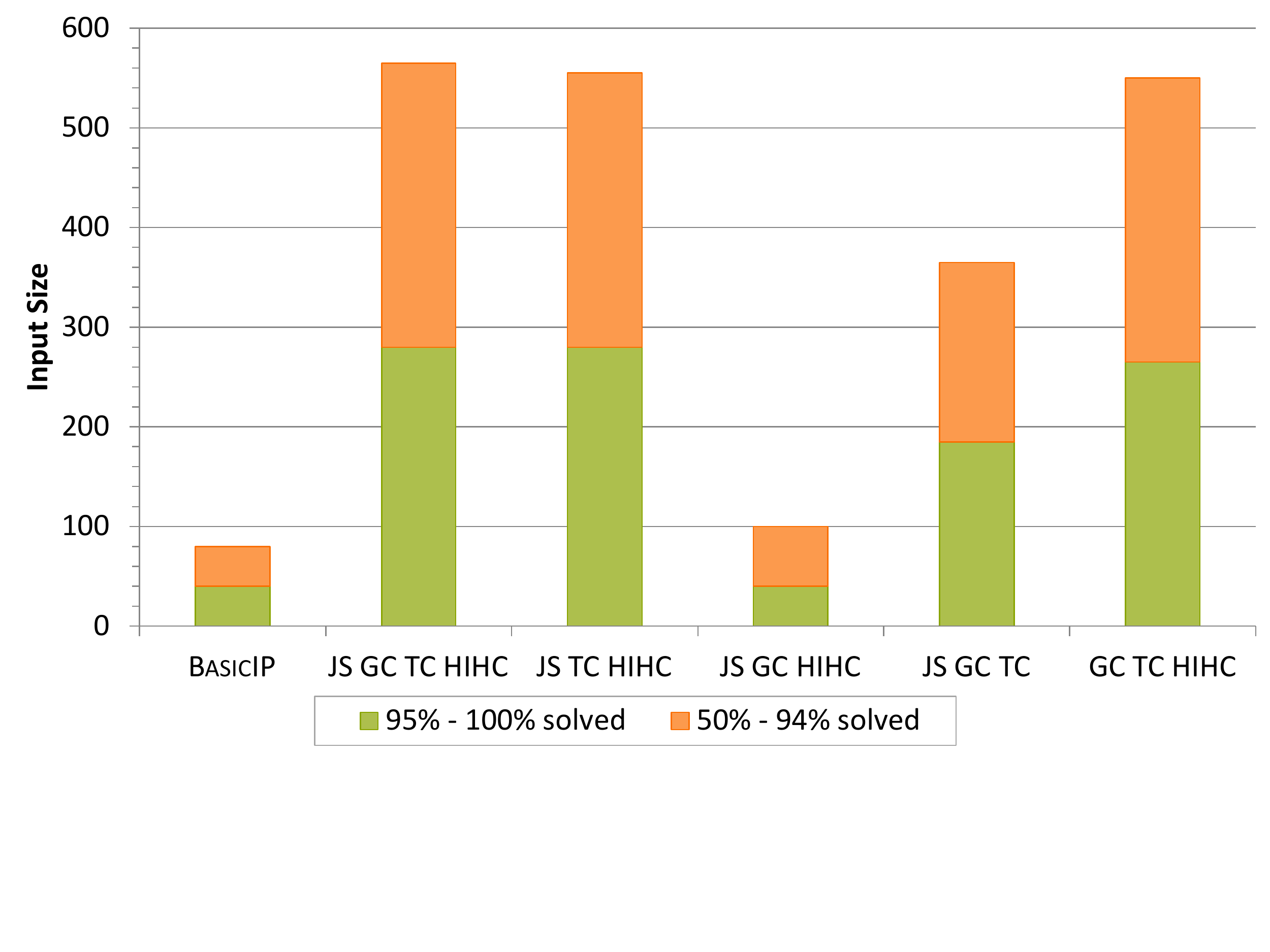}
		\includegraphics[width=0.48\textwidth, trim= 0mm 40mm 0mm 0mm, clip]{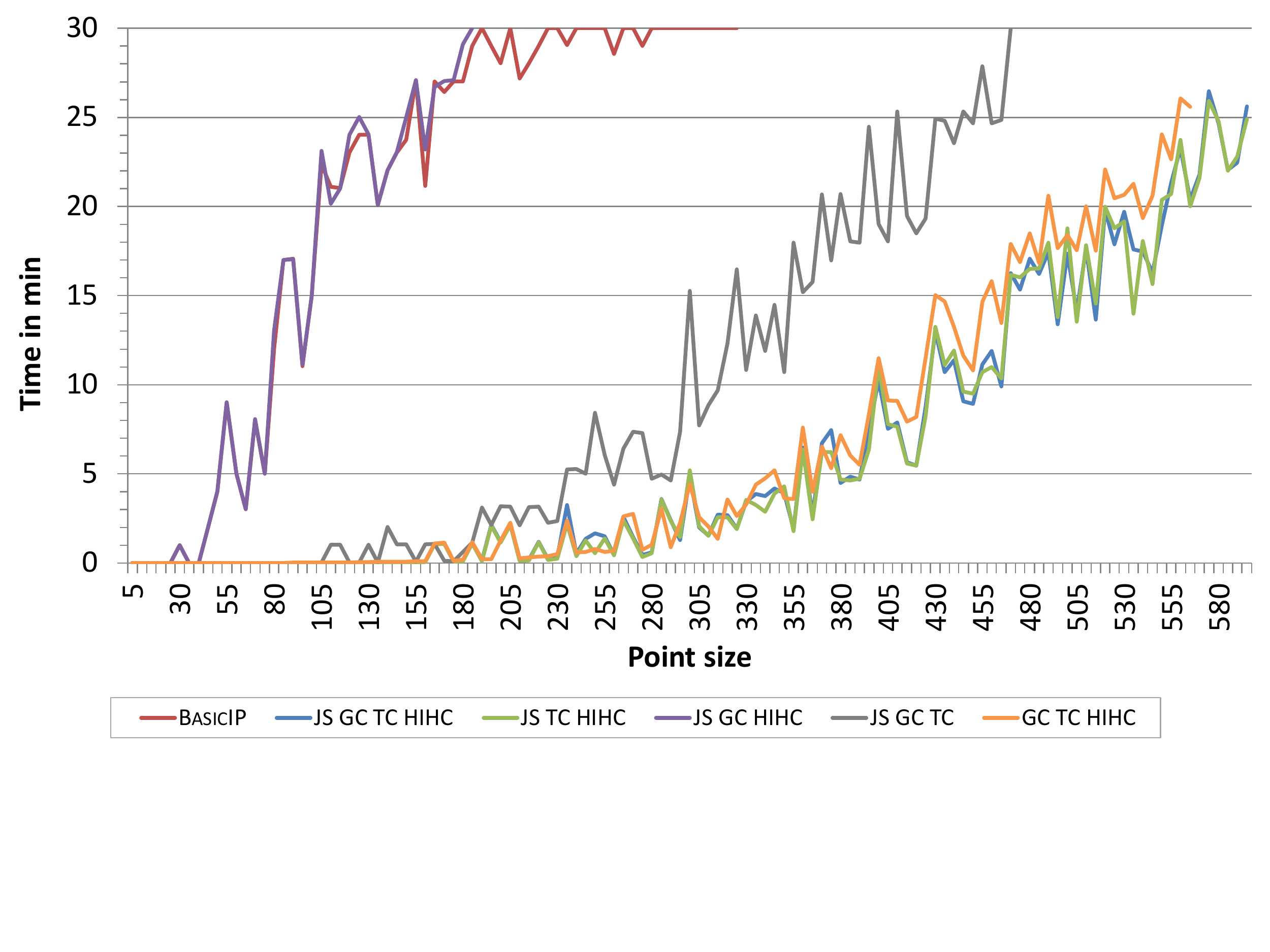}
		\caption{(Left) Success rate for the different variants of using of the cuts, with 30 instances for each input size ($y$-axis).
(Right) The average runtime of the different variants for all 30 instances. A non-solved instance is interpreted as 30 minutes runtime.}
		\label{fig:diagInst}
	\end{center}
\end{figure}

\old{
\begin{figure}
	\begin{center}
		\includegraphics[width=0.75\textwidth, trim= 0mm 40mm 0mm 0mm, clip]{fig/DiagrammRuntime.pdf}
		\caption{The average runtime of the different variants for all 30 instances. A non-solved instance is interpreted as 30 minutes runtime.}
		\label{fig:diagRuntime}
	\end{center}
\end{figure}
}

\begin{figure}
	\begin{center}
		\includegraphics[width=0.48\textwidth, trim= 0mm 40mm 0mm 0mm, clip]{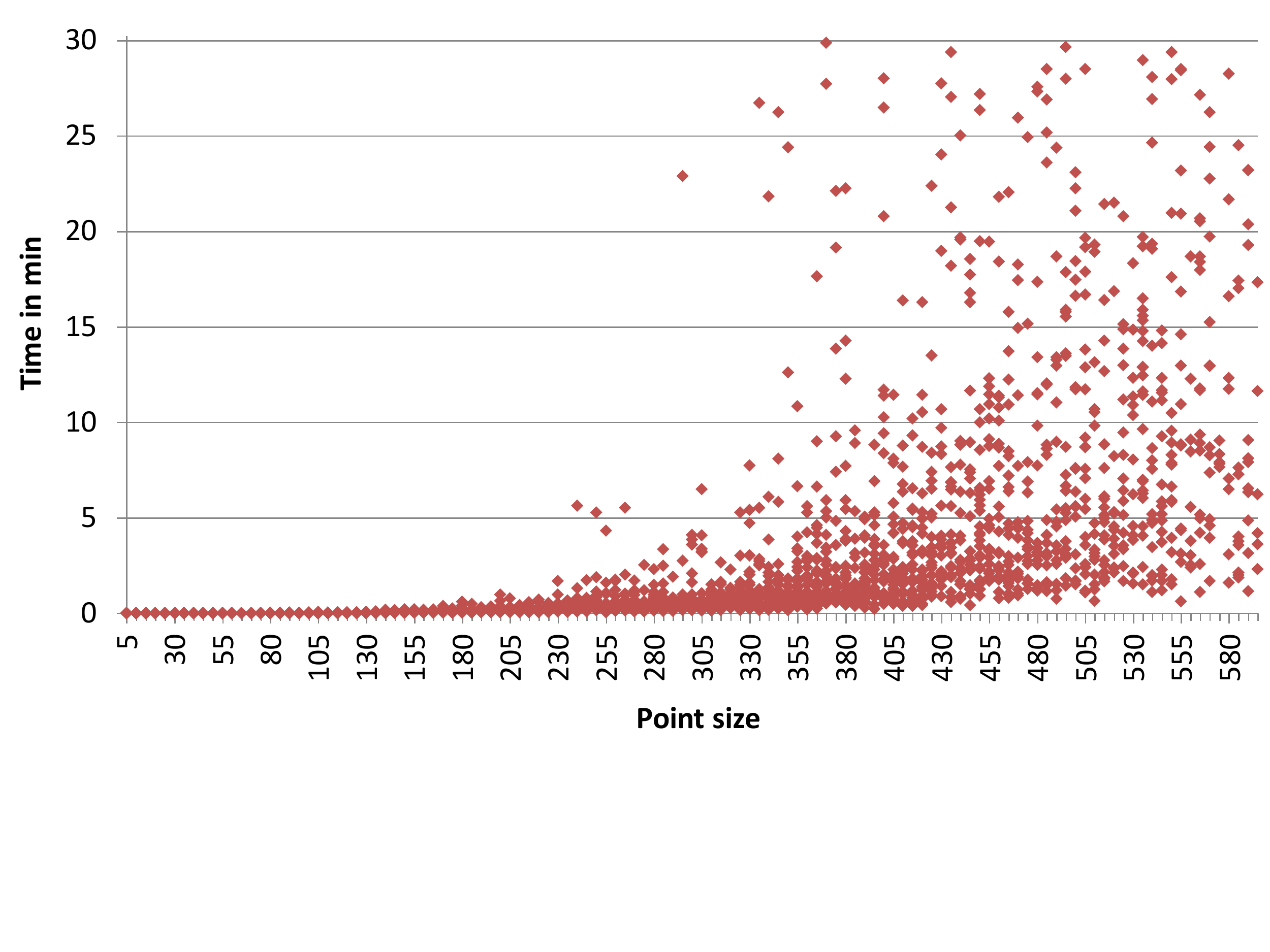}
		\includegraphics[width=0.48\textwidth, trim= 0mm 40mm 0mm 0mm, clip]{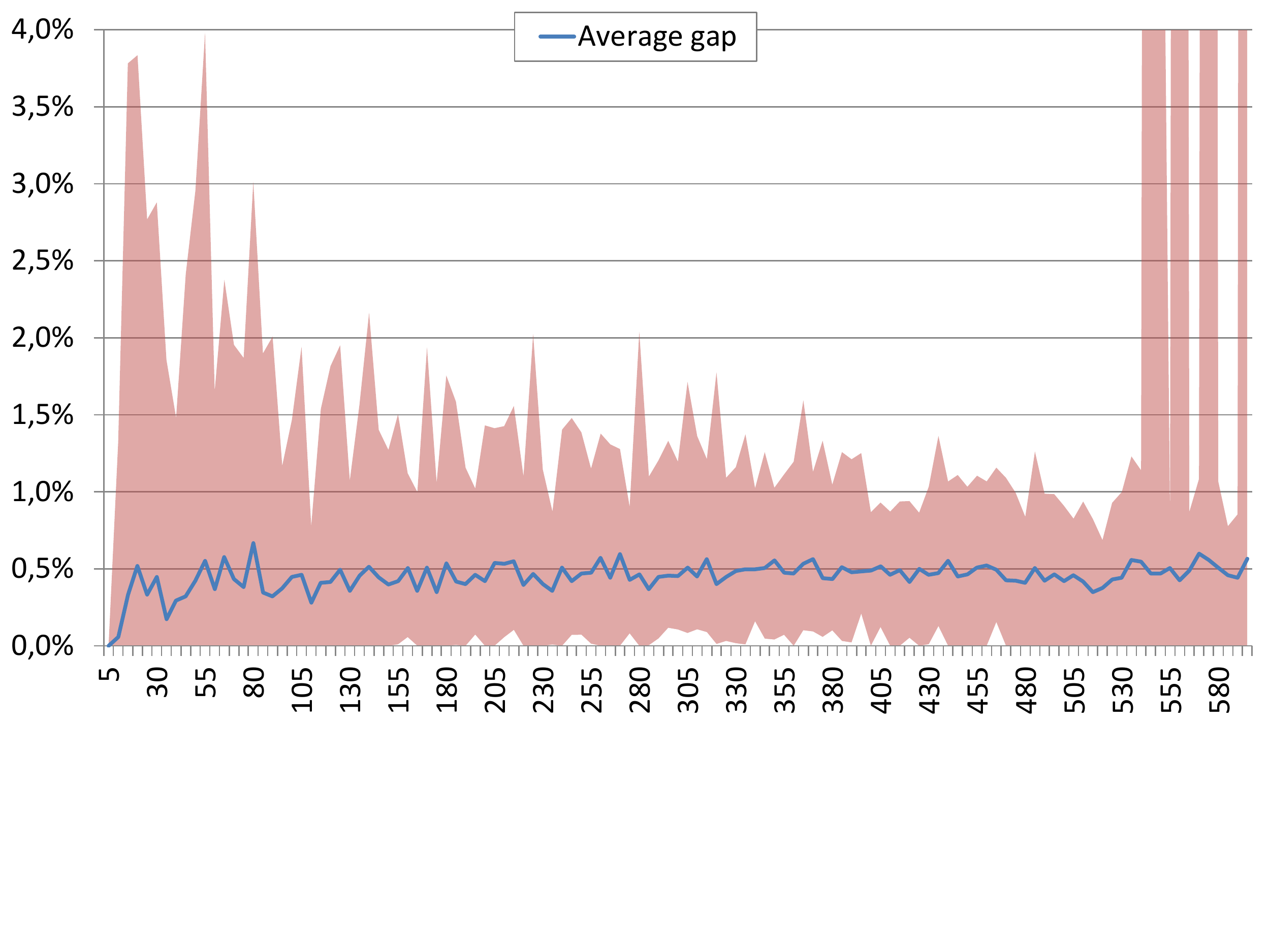}
		\caption{(Left) The distribution of the runtime within 30 minutes for the case of using the jumpstart, glue cuts, tail cuts and HiH-cuts.
		(Right) The relative gap of the value on the edges of the Delaunay triangulation to the optimal value. 
The red area marks the range between the minimal and maximal gap.}
		\label{fig:diagDistribution}
	\end{center}
\end{figure}

\begin{figure}
	\begin{center}
		\includegraphics[width=0.75\textwidth, trim= 0mm 40mm 0mm 0mm, clip]{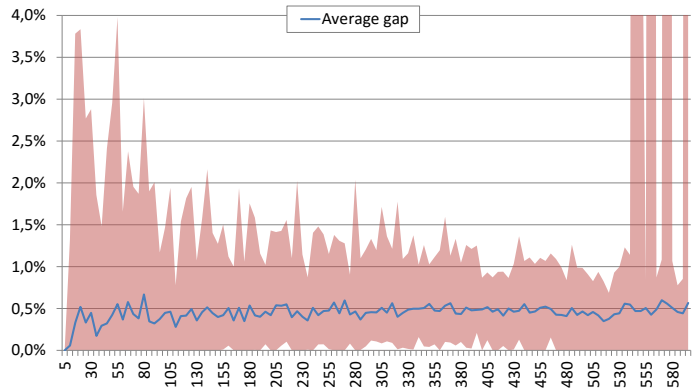}
		\caption{The relative gap of the value on the edges of the Delaunay triangulation to the optimal value. The red area marks the range between the minimal and maximal gap.}
		\label{fig:diagGapOptJS}
	\end{center}
\end{figure}

\begin{figure}
	\subfigure[Earth by night \label{fig:earthByNight}] {\includegraphics[width=0.48\columnwidth]{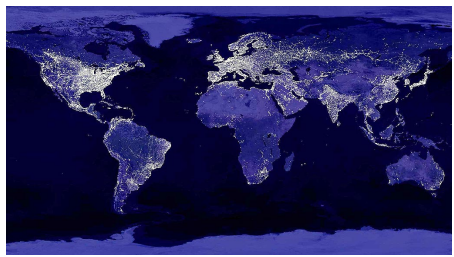}} 
	\hfill
	\subfigure[A sampled instance]{\includegraphics[width=0.48\columnwidth]{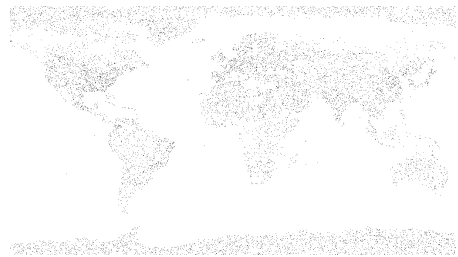}}
	\caption{Using a brightness map as a density function for generating clustered point sets.}
	\label{fig::generator}
\end{figure}

We observe that even without using glue cut and jumpstart, 
we are able to solve more than $50 \%$ of the
instances up to about 550 input points. Without the tail cuts, 
we hit a wall at 100 points, without the HiH-cut instances,
at about 370 input points; see Figure~\ref{fig:diagInst},
which also shows the average runtime of all 30 instances for all
variants. Instances exceeding the 30 minutes time limit are 
marked with a 30-minutes timestamp. The figure shows that using jumpstart shortens the runtime
significantly; using the glue cut is almost as fast as the variant
without the glue cut. 

Figure~\ref{fig:diagDistribution} shows that medium-sized instances 
(up to about 450 points) can be solved in under 5 minutes.
We also show 
that restricting the
edge set to the Delaunay triangulation edges yields solutions that are about
0.5\% worse on average than the optimal solution. Generally the solution of the
jumpstart gets very close to the optimal solution until about 530 points. After
that, for some larger instances, we get solutions on the edge set of the
Delaunay triangulation that are up to $50 \%$ worse than the optimal solution. 

\old{
\subsection{Clustered point sets}
So far we only considered random instances with an uniform distribution.
 As a source for instances possibly having clustered point sets, we have written a small tool to generate such instances from an arbitrary image. The brightness values of an image induce a density function that can be used for sampling (see Figure \ref{fig::generator}). To reduce noise, we cut off brightness values below a certain threshold, i.e., we set the probability to choose the resp. pixels to zero. 

\begin{figure}
	\subfigure[Earth by night \label{fig:earthByNight}] {\includegraphics[width=0.45\columnwidth]{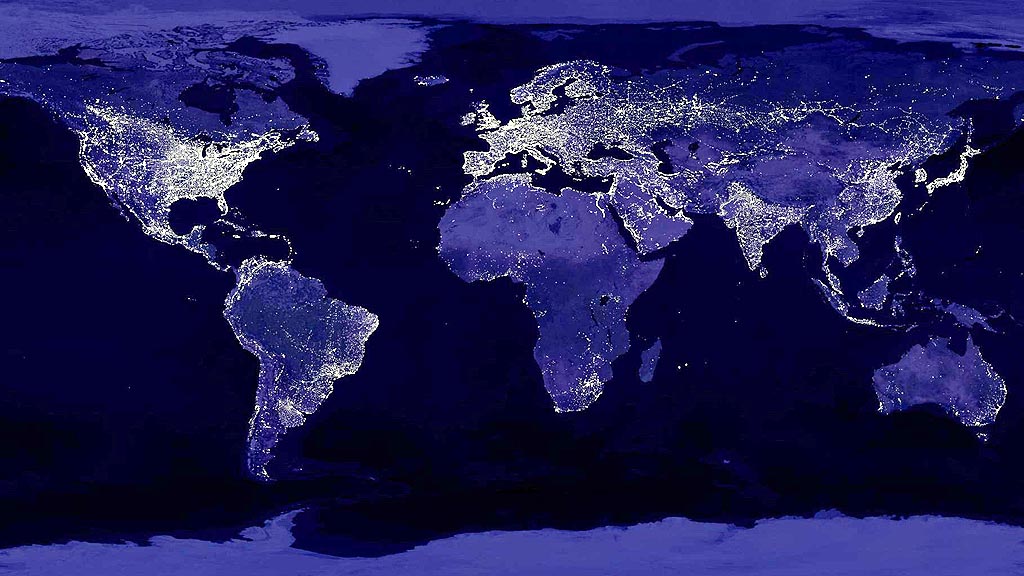}} 
	\hfill
	\subfigure[A sampled instance]{\includegraphics[width=0.45\columnwidth]{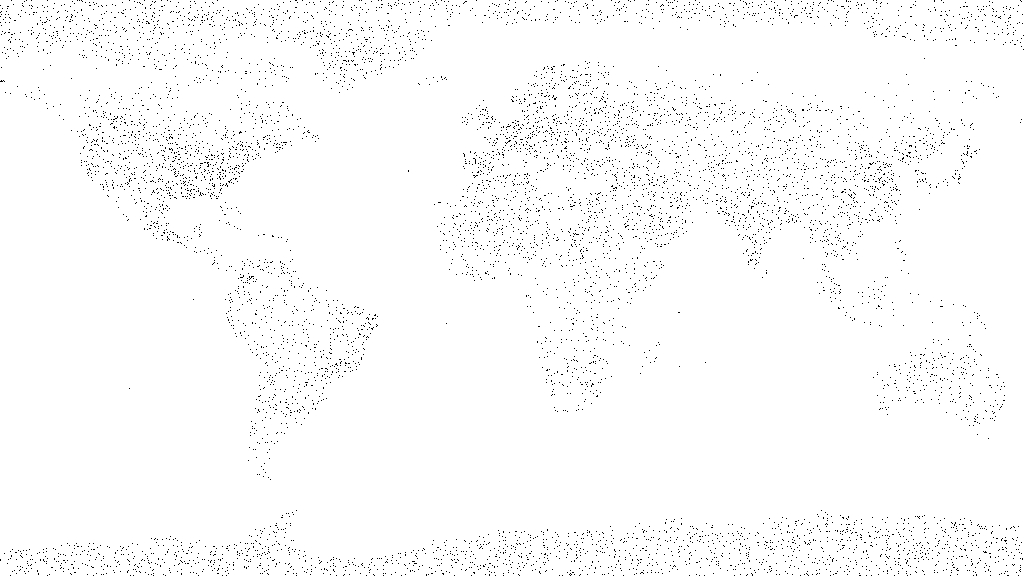}}
	\caption{Using the brightness of an image as a density function to generate clustered point sets.}
	\label{fig::generator}
\end{figure}

\ignore{\todo[inline]{Write out the following:}
\begin{itemize}
	\item 10 instances per point size, start 200 points, increased by 20 point until 1000.
	\item Instances generated by using Figure \ref{fig:earthByNight} as source, only bright areas used => instances are very similar and strongly clustered.
	\item Until 600 points almost all instances are solved. $>600$ only a few.
	\item Largest solved instance with 960 points. See Figure \ref{fig:960instance}.	
	\item TSPLib used, too. Largest instance gr666.tsp. All instances seen as they are in Euclidean space, not in their given geometry space. See the full version for details.
\end{itemize}}

Based on the satellite image from Figure \ref{fig:earthByNight}, we sampled point sets where points are only set at bright areas. This yields strongly clustered instances.
The point set sizes start at 200 points and increase by 20 points until we reach a total size of 1000 points. For each point set size we generated 10 instances.

Figure \ref{fig:diagSatellite} shows that up to a point set size of 600, almost all instances can be solved within 30 minutes. When dealing with 760 points and above, we can only solve up to $50\%$ of the instances. The last instance solved has a size of 960 points (see Figure \ref{fig:960instance}).
}

\old{
\begin{figure}
	\begin{center}
		\includegraphics[width=0.75\textwidth, trim= 0mm 45mm 0mm 0mm, clip]{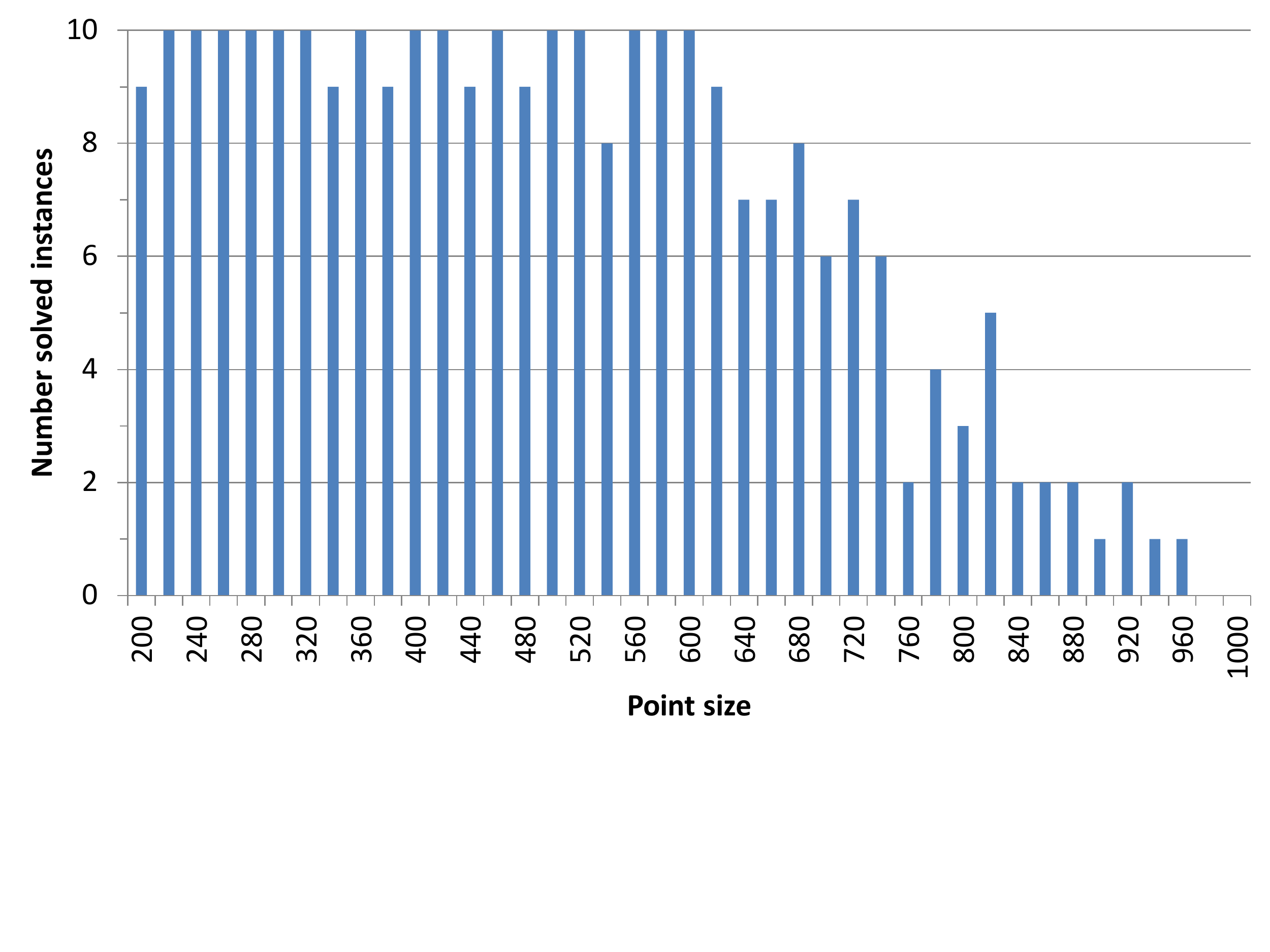}
		\caption{The number of instances which are solvable within 30 minutes. Per point size are 10 instances used.}
		\label{fig:diagSatellite}
	\end{center}
\end{figure}
}


\section{Conclusions}
\label{sec:conclusion}

As discussed in the introduction, considering general instead of simple
polygons corresponds to searching 
for a shortest cycle cover with a specific topological constraint:
one outside cycle surrounds a set of disjoint and unnested inner cycles. 
Clearly, this is only one example of
considering specific topological constraints. 
Our techniques and results should be applicable, after suitable
adjustments, to other constraints on the topology of cycles.
\old{Our geometric techniques for all aspects of the problem (complexity, approximation,
exact computation) should still be applicable, after suitable adjustments.}
We gave a 3-approximation for \MinBound; we expect that \MinBound has
a polynomial-time approximation scheme, base on PTAS
techniques~\cite{arora,mitchell} for geometric TSP, and we will
elaborate on this in the full paper.
\old{More specific aspects can also be expanded. We are optimistic that the
guillotine subdivision techniques by Mitchell can be applied to achieve a 
polynomial-time approximation scheme; as the outcome promises to be more
theoretical than practical, we have refrained from elaborating these ideas
at this time.}

There are also various practical aspects that can be explored further.
It will be interesting to evaluate the practical performance of the theoretical
approximation algorithm, not only from a practical perspective, but also
to gain some insight on whether the approximation factor of 3 can be tightened.
Pushing the limits of solvability can also be attempted, e.g., by using more advanced
techniques from the TSP context.
We can also consider sparsification techniques other than the Delaunay edges,
e.g., the union between the best known tour and the $k$-nearest-neighbor edge set $(k
\in \{2,5,10,20\})$ has been applied for TSP by Land \cite{land1979solution},
or (see Padberg and Rinaldi \cite{padberg1991branch}) by taking the union of $k$ tours acquired by Lin's and Kernighan's
heuristic algorithm \cite{lin1973effective}.

\section*{Acknowledgements}
We thank Stephan Friedrichs and Melanie Papenberg for helpful conversations.
Parts of this work were carried out at the 30th Bellairs Winter Workshop on Computational Geometry (Barbados) in 2015.
We thank the workshop participants and organizers, particularly Erik Demaine. 
Joseph~Mitchell is partially supported by NSF (CCF-1526406).
Irina~Kostitsyna is supported by the Netherlands Organisation for Scientific Research (NWO) under project no. 639.023.208.

\small 
\bibliographystyle{abbrv}

\bibliography{bibliography}

\newcommand{\Berg}{de Berg}
\begin{thebibliography}{10}

\bibitem{altmehl}
E.~Althaus and K.~Mehlhorn.
\newblock {T}raveling {S}alesman-based curve reconstruction in polynomial time.
\newblock {\em {SIAM} J. Comput.}, 31(1):27--66, 2001.

\bibitem{abcc}
D.~L. Applegate, R.~E. Bixby, V.~Chvatal, and W.~J. Cook.
\newblock On the solution of {T}raveling {S}alesman {P}roblems.
\newblock {\em Documenta Mathematica -- Journal der Deutschen
  Mathematiker-Vereinigung, ICM}, pages 645--656, 1998.

\bibitem{tspstudy}
D.~L. Applegate, R.~E. Bixby, V.~Chvatal, and W.~J. Cook.
\newblock {\em The Traveling Salesman Problem: A Computational Study (Princeton
  Series in Applied Mathematics)}.
\newblock Princeton University Press, Princeton, NJ, USA, 2007.

\bibitem{arora}
S.~Arora.
\newblock Polynomial time approximation schemes for {E}uclidean {T}raveling
  {S}alesman and other geometric problems.
\newblock {\em J. {ACM}}, 45(5):753--782, 1998.

\bibitem{chew1989constrained}
L.~P. Chew.
\newblock Constrained {D}elaunay triangulations.
\newblock {\em Algorithmica}, 4(1-4):97--108, 1989.

\bibitem{christofides}
N.~Christofides.
\newblock Worst-case analysis of a new heuristic for the {T}ravelling
  {S}alesman {P}roblem.
\newblock Technical Report Report 388, Graduate School of Industrial
  Administration, CMU, 1976.

\bibitem{tsppursuit}
W.~J. Cook.
\newblock {\em In Pursuit of the Traveling Salesman: Mathematics at the Limits
  of Computation}.
\newblock Princeton University Press, Princeton, NJ, USA, 2012.

\bibitem{bills-book}
W.~J. Cook, W.~H. Cunningham, W.~R. Pulleyblank, and A.~Schrijver.
\newblock {\em Combinatorial Optimization}.
\newblock Wiley, 1998.

\bibitem{dots}
T.~K. Dey, K.~Mehlhorn, and E.~A. Ramos.
\newblock Curve reconstruction: Connecting dots with good reason.
\newblock {\em Comput. Geom.}, 15(4):229--244, 2000.

\bibitem{dillencourt1987non}
M.~B. Dillencourt.
\newblock A non-hamiltonian, nondegenerate {D}elaunay triangulation.
\newblock {\em Information Processing Letters}, 25(3):149--151, 1987.

\bibitem{bff+-corsct-15}
S.~P. Fekete, S.~Friedrichs, M.~Hemmer, M.~Papenberg, A.~Schmidt, and
  J.~Troegel.
\newblock Area- and boundary-optimal polygonalization of planar point sets.
\newblock In {\em EuroCG 2015}, pages 133--136, 2015.

\bibitem{giesen}
J.~Giesen.
\newblock Curve reconstruction, the {Traveling Salesman Problem} and {M}enger's
  theorem on length.
\newblock In {\em Proc. 15th Annual Symp. Comp. Geom. (SoCG)}, pages 207--216,
  1999.

\bibitem{Groetschel1980a}
M.~Gr{\"o}tschel.
\newblock On the symmetric {T}ravelling {S}alesman {P}roblem: solution of a
  120-city problem.
\newblock {\em Mathematical Programming Study}, 12:61--77, 1980.

\bibitem{tsp-book2}
G.~Gutin and A.~P. Punnen.
\newblock {\em The Traveling Salesman Problem and Its Variations}.
\newblock Springer, 2007.

\bibitem{junger1995traveling}
M.~J{\"u}nger, G.~Reinelt, and G.~Rinaldi.
\newblock The {T}raveling {S}alesman {P}roblem.
\newblock {\em Handbooks in Operations Research and Management Science},
  7:225--330, 1995.

\bibitem{land1979solution}
A.~Land.
\newblock The solution of some 100-city {T}ravelling {S}alesman {P}roblems.
\newblock Technical report, London School of Economics, 1979.

\bibitem{tsp-book}
E.~L. Lawler, E.~L. Lawler, and A.~H. Rinnooy-Kan.
\newblock {\em The Traveling Salesman Problem: A Guided Tour of Combinatorial
  Optimization}.
\newblock Wiley, 1985.

\bibitem{lin1973effective}
S.~Lin and B.~W. Kernighan.
\newblock An effective heuristic algorithm for the {T}raveling-{S}alesman
  problem.
\newblock {\em Operations research}, 21(2):498--516, 1973.

\bibitem{mitchell}
J.~S.~B. Mitchell.
\newblock Guillotine subdivisions approximate polygonal subdivisions: {A}
  simple polynomial-time approximation scheme for geometric {TSP}, k-{MST}, and
  related problems.
\newblock {\em {SIAM} J. Comput.}, 28(4):1298--1309, 1999.

\bibitem{padberg1991branch}
M.~Padberg and G.~Rinaldi.
\newblock A branch-and-cut algorithm for the resolution of large-scale
  symmetric {T}raveling {S}alesman {P}roblems.
\newblock {\em SIAM Rev.}, 33(1):60--100, 1991.

\bibitem{pferschy2014generating}
U.~Pferschy and R.~Stanek.
\newblock Generating subtour constraints for the {TSP} from pure integer
  solutions.
\newblock {\em Department of Statistics and Operations Research, University of
  Graz, Tech. Rep}, 2014.

\bibitem{reinelt1991tsplib}
G.~Reinelt.
\newblock {TSPlib} -- {A} {T}raveling {S}alesman {P}roblem library.
\newblock {\em ORSA J. on Computing}, 3(4):376--384, 1991.

\bibitem{cgal:arr}
R.~Wein, E.~Berberich, E.~Fogel, D.~Halperin, M.~Hemmer, O.~Salzman, and
  B.~Zukerman.
\newblock {2D} arrangements.
\newblock In {\em {CGAL} User and Reference Manual}. {CGAL Editorial Board},
  {4.3} edition, 2014.

\end{thebibliography}


\end{document}